\documentclass[11pt]{llncs}

\usepackage{amssymb}
\usepackage{amsmath}
\usepackage{amsthm}
\usepackage{textcomp}
\usepackage{array}
\usepackage{lineno}
\usepackage{multirow}
\usepackage[latin5]{inputenc}
\usepackage{color}
\usepackage{setspace}
\usepackage[left=2.5cm,top=3cm,right=2.5cm,bottom=3cm]{geometry}

\theoremstyle{definition}

\newtheorem{fact}{Fact}[section]
\newtheorem{openproblem}{Open Problem}

\newcommand{\ket}[1]{|#1\rangle}

\newcommand{\cent}[0]{\mbox{\textcent}}
\newcommand{\dollar}[0]{\$}

\title{Languages recognized by nondeterministic quantum finite automata\thanks{This work was partially 
supported by the Scientific and Technological Research Council of Turkey (T\"{U}B\.ITAK) with grant 108142
and the Bo\~{g}azi\c{c}i University Research Fund with grant 08A102.}}
\author{Abuzer Yakary{\i}lmaz\ \and A.C. Cem Say }
\institute{Bo\u{g}azi\c{c}i University, Department of Computer Engineering,\\ Bebek 34342 \.{I}stanbul, Turkey \\
\email{{abuzer,say}@boun.edu.tr}
 \\~~\\
February 12, 2010
}

\begin{document}

\newlength{\twidth}
\maketitle
\pagenumbering{arabic}

\begin{abstract} \label{abstract:Abstract}

The nondeterministic quantum finite automaton (NQFA) is the only known
case where a one-way quantum finite automaton (QFA) model has been
shown to be strictly superior in terms of language recognition power
to its probabilistic counterpart. We give a characterization of the
class of languages recognized by NQFA's, demonstrating that it is
equal to the class of exclusive stochastic languages. We also
characterize the class of languages that are recognized necessarily by
two-sided error by QFA's. It is shown that these classes remain the
same when the QFA's used in their definitions are replaced by several
different model variants that have appeared in the literature. We
prove several closure properties of the related classes.
The ramifications of these results about classical and quantum
sublogarithmic space complexity classes are examined.

\end{abstract}

\section{Introduction} \label{section:Introduction}
An interesting feature of both probabilistic and quantum computational
models is that in some cases, the set of problems that can be solved
gets larger when the automaton in question is allowed to make more
error in its decisions, whereas in some other cases, such a relaxation
does not increase the computational power at all. When one-way
probabilistic finite automata (PFA's) are required to make no error in
their decisions, they recognize exactly the class of regular
languages. When they are allowed to make \textit{bounded error}, that is, to
give the correct response for each input with probability at least $ \frac{1}{2} + \delta $,
for a fixed $ \delta >0 $, the class of languages
that are recognized remains the same. The computational power of PFA's
is seen to increase only when we allow two-sided \textit{unbounded error},
where the only requirement is that all members of the recognized
language are accepted
with probability greater than the acceptance probability of any nonmember.

There are several alternative models of quantum finite automata
(QFA's), and differences  (e.g. in features regarding the form and
number of measurements that can be performed on the machine, whether
an output tape is included or not, etc.) in their definitions, which
would not affect the power of classical automata, yield QFA variants
of differing computational power. In the zero-error and bounded-error
cases, some QFA variants \cite{MC00,KW97,Na99} are
strictly inferior to the corresponding PFA's from the point of view of
language recognition power, whereas the most general models
\cite{Pa00,Ci01,BMP03,Hi07} are
equivalent to their probabilistic counterparts in those settings. In
the unbounded error case, the languages
recognized with cutpoint by the weakest QFA model \cite{MC00} form a
proper subclass
\cite{BC01B}
of the corresponding classical class (the stochastic languages),
whereas it was discovered recently \cite{YS09A,YS09D} that more generalized QFA
variants, including the popular Kondacs-Watrous model, are equivalent
to PFA's in this case as well.

With regard to state complexity, sufficiently general QFA models can
simulate all zero-error and bounded-error PFA's with small overhead,
and some regular languages
have bounded-error
QFA's that are exponentially smaller than the corresponding PFA
\cite{AF98}. In the two-sided unbounded error setting, quantum and
probabilistic
machines can simulate each other with only a polynomial overhead in
the number of states \cite{YS09D}.

We study the computational power of QFA's in the
one-sided unbounded error setting, where one of the two responses that
the machine can output about the membership of the input string in the
recognized language is correct with certainty, and the other response
has a nonzero probability of being correct. Since the error bound can
be improved by repeating the computation, an examination of languages
recognizable
in this setting is significant for understanding the power of
generalizations of the
underlying model to include, say, a two-way tape head. Just like their
classical counterparts,
QFA's that recognize their languages with cutpoint 0, (that is, with positive
one-sided error,) are also known as \textit{nondeterministic} machines.
It is well known that classical nondeterministic finite automata
recognize precisely the regular languages.
In notable previous work on nondeterministic quantum finite automata (NQFA's),
Bertoni and Carpentieri have shown \cite{BC01B} that the class of
languages recognized by
NQFA's of the Moore-Crutchfield type does not contain any nonempty
finite languages,
but does contain the nonregular language $ L_{neq}=\{w\in\{a,b\}^{*}
\mid |w|_{a} \neq |w|_{b}\} $, where $ |w|_{\sigma} $
denotes the number of occurrences of the symbol $ \sigma $  in the string $ w $.
Nakanishi et al. \cite{NIHK02} considered the somewhat more powerful
Kondacs-Watrous model of QFA's,
and proved that NQFA's of this type can recognize all regular
languages, establishing their
strict superiority over their classical counterparts.

In this paper, we give a full characterization of the class of languages
recognized by all NQFA variants that are at least as general as the
Kondacs-Watrous type, demonstrating that it is
equal to the class of exclusive stochastic languages. This lies
properly between the classes of
languages recognized with zero error and two-sided unbounded error by
QFA's\footnote{To our knowledge, this is the only case
where these three classes have been shown to be distinct for any
automaton model, be it quantum or classical.}.  
Every regular language has a NQFA with at most linearly more states than the
corresponding classical nondeterministic finite automaton (NFA), and
there exist infinite families of regular languages which can be
recognized by just tuning the transition amplitudes of a NQFA with a
constant number of states, whereas the sizes of the corresponding
NFA's grow without bound. We also prove several new closure properties
of the related classes, and examine what these results imply about
the comparative power of probabilistic vs. quantum Turing machines
with small space bounds.

The rest of this paper is structured as follows: Section 2 contains
the relevant definitions and previously known facts. In Section 3, we give a
characterization of the class of languages recognized by
Kondacs-Watrous NQFA's, and discuss the superiority of
several NQFA variants over their classical counterparts
in terms of language recognition and succinctness. An examination of
the relationships among languages which can be
recognized by QFA's with one-sided error and those that require two-sided error
is presented in Section 4. Section 5 contains several proofs of
closure properties for the classes of languages recognized with
one-sided error. Section 6 is a conclusion,
where we examine the consequences of the NQFA results for classical
and quantum sublogarithmic
space complexity classes.

\section{Preliminaries} \label{section:Preliminaries}
\subsection{Automata}
In the following, $ \Sigma $ denotes the input alphabet, not
containing the end-markers $ \cent $ and $ \dollar $, and
$ \Gamma  $ is the tape alphabet, such that $ \Gamma = \Sigma \cup
\{\cent,\dollar\} $.
\\
\begin{definition}
\label{definition:1pfa}
A \textit{(1-way) probabilistic finite automaton} (PFA) with $ n \in
\mathbb{Z}^{+} $ states is a 4-tuple
$ \mathcal{P}=(S,\Sigma,\{\mathsf{A}_{\sigma \in \Gamma} \},F) $, where
\begin{enumerate}
    \item $ S  = \{s_1, \cdots, s_n \} $ is the set of states, and $
s_{1} $ is the start state,
    \item $ \mathsf{A}_{\sigma} $ is the $ n \times n $ real-valued
stochastic transition matrix
    for symbol $ \sigma $, that is, $ \mathsf{A}_{\sigma}(i,j) $ is
the value of the transition
    probability from state $ s_{i} $ to state $ s_{j} $ when reading
symbol $ \sigma $,
    \item  $F \subseteq S$ is the set of accepting states.
\end{enumerate}
\end{definition}

The probability distribution of $ \mathcal{P} $'s states at any point
during the processing of the input string can be traced using an
$n$-element row vector. For an input string $ w \in \Sigma^{*} $, $
\mathsf{w}=\cent w \dollar $,
$ \mathsf{v}_{0}  =  (1,0,\cdots,0)_{1 \times n } $ denotes the
initial state vector.
The effect of reading the $ i $th tape symbol can be calculated by
multiplying the vector $ \mathsf{v}_{i-1} $  by
the matrix $ \mathsf{A}_{\mathsf{w}_{i}} $,  yielding $ \mathsf{v}_{i} $.
$ \mathsf{v}_{| \mathsf{w} |} = \mathsf{v}_{0}
\mathsf{A}_{\mathsf{w}_{1}} \cdots
\mathsf{A}_{\mathsf{w}_{| \mathsf{w} |}} $ denotes the final state vector.
The acceptance probability of $ w $ by $ \mathcal{P} $ is
\begin{equation}
    \label{pfa:f_P_w}
    f_{\mathcal{P}}(w) = \sum_{s_{i} \in F} \mathsf{v}_{|\mathsf{w}|}(i),
\end{equation}
where $ \mathsf{v}_{|\mathsf{w}|}(i) $ denotes the $ i $th entry
of $ \mathsf{v}_{|\mathsf{w}|} $.
\\
\begin{definition}
\label{definition:gpfa}
A \textit{generalized probabilistic finite automaton} (GPFA) with $ n
\in \mathbb{Z}^{+} $ states is a 5-tuple
$ \mathcal{G}=(S,\Sigma,\{\mathsf{A}_{\sigma \in
\Sigma}\},\mathsf{v}_{0},\mathsf{f}) $, where
\begin{enumerate}
    \item $ S=\{s_{1},\cdots,s_{n}\} $ is the set of states,
    \item $ \mathsf{A}_{\sigma} $ is the $ n \times n $ real-valued
transition matrix for symbol $ \sigma $,
    that is, $\mathsf{A}_{\sigma}(i, j)$ is the (possibly negative)
``weight" of the transition from state
    $ s_i $ to state $ s_j $ when reading symbol $ \sigma$,
    \item $ \mathsf{v}_{0} $ is the real-valued initial $ 1 \times
n $ vector, and,
    \item  $ \mathsf{f} $ is the real-valued final $ n \times 1 $ vector.
\end{enumerate}
\end{definition}

A GPFA $ \mathcal{G} $ is associated by a function $f_{\mathcal{G}} :
\Sigma^{*} \rightarrow \mathbb{R} $,
in the following way: For an input string $ w \in \Sigma^{*} $,
\begin{eqnarray}
    \label{equation:gpfa:f_G_w}
    f_{\mathcal{G}}(w) & = &
\mathsf{v_{0}}\mathsf{A}_{w_{1}}\cdots\mathsf{A}_{w_{|w|}} \mathsf{f}.
\end{eqnarray}
\\
\begin{definition}
\label{definition:kwqfa}
A \textit{(1-way) Kondacs-Watrous quantum finite automaton}
(KWQFA) \cite{KW97} with $ n \in \mathbb{Z}^{+} $ states is a 5-tuple
$ \mathcal{M}=(Q, \Sigma, \{\mathsf{U}_{\sigma \in \Gamma} \},Q_{acc},
Q_{rej} ) $, where
\begin{enumerate}
    \item $ Q $ = $ \{q_1, \cdots q_n \} $ is the set of states,
and $ q_{1} $ is the initial state,
    \item $ \mathsf{U}_{\sigma} $ is the $ n \times n $
complex-valued unitary transition matrix for symbol
             $ \sigma $, that is, $ \mathsf{U}_{\sigma}(j,i) $
             is the amplitude of the transition from $ q_{i} $ to $
q_{j} $ when
reading the symbol $ \sigma $,
    \item $ Q_{acc} $ and $ Q_{rej} $, disjoint subsets of $ Q $,
are the sets of accepting and rejecting states,
    and $ Q_{non} = Q \setminus (Q_{acc} \cup Q_{rej}) $ is the set
of non-halting states.
\end{enumerate}
\end{definition}

The amplitude distribution of the states of a quantum automaton is
represented by an $n$-element column vector.
$ \ket{\mathsf{u}­_{0}} $, the initial state vector, equals $
(1,0,\cdots,0)^{\mathtt{T}}_{1 \times n } $.
Note the difference with probabilistic automata.

For a given input string $ w \in \Sigma^{*} $, $ \mathcal{M} $ scans the tape,
containing $ \mathsf{w}=\cent w \dollar $, from the left to the right.
During the processing of each symbol, the
machine undergoes two operations: First, its state vector evolves
according to the unitary
transformation associated with the scanned symbol, that is,
\[ \ket{u_{i}} = \mathsf{U}_{w_{i}} \ket{u_{i-1}}. \]
Then, the machine is observed to
see whether it has accepted, rejected, or not
halted yet. At this point, each accepting state with amplitude $ \alpha $ adds
$ |\alpha |^{2} $
to the overall acceptance probability $ f_{\mathcal{M}}(w) $ of the
input\footnote{This is the behavior
allowed by the general KWQFA definition.}. In
the particular KWQFA's that will be described
in this paper, intermediate observations can yield the result ``reject", but
the accepting states can be entered only at the end of the
computation, after scanning
the right end-marker $ \dollar $.
Halting states ``drop out" of the state vector $ \ket{u}_{i} $, their
amplitudes being
replaced with zeros, and the head moves on to the next symbol.

As we have defined them, PFA's process all of the input string before
deciding on acceptance or rejection,
whereas KWQFA's can halt before reaching the end of the input. (The
QFA variant that precisely corresponds to
Definition \ref{definition:1pfa} is the \textit{Moore-Crutchfield QFA}
(MCQFA) \cite{MC00}.)
This difference should not distract the reader, since it is easy to
show that the classes of
languages recognized by PFA's, both with general cutpoint, and with
cutpoint 0, (to be defined in the next
subsection,) do not change when the model is modified to give it this
additional capability\footnote{Note
that, when comparing two QFA variants with each other, this kind of
difference is very important,
since it usually affects the computational power of the models.}.
This is true for all PFA variants that may be obtained by appropriately
reconfiguring Definition \ref{definition:1pfa} to correspond to the various QFA
models that are cited in this paper. The only crucial distinction
between Definitions \ref{definition:1pfa} and \ref{definition:kwqfa}
is the one between classical and quantum.

\subsection{Languages}
~~

\begin{definition}
     An automaton $ \mathcal{A} $ defined over alphabet $ \Sigma $ divides
$ \Sigma^{*} $ into three disjoint subsets
     with cutpoint $ \lambda \in \mathbb{R} $:
     \begin{enumerate}
             \item $ \mathbb{L}(\mathcal{A},<\lambda) = \{w \in
\Sigma^{*} \mid
f_{\mathcal{A}}(w)<\lambda \} $,
             \item $ \mathbb{L}(\mathcal{A},=\lambda) = \{w \in
\Sigma^{*} \mid
f_{\mathcal{A}}(w)=\lambda \} $,
             \item $ \mathbb{L}(\mathcal{A},>\lambda) = \{w \in
\Sigma^{*} \mid
f_{\mathcal{A}}(w)>\lambda \} $.
     \end{enumerate}
     Additionally, we define $ \mathbb{L}(\mathcal{A},\neq\lambda) =
\mathbb{L}(\mathcal{A},<\lambda) \cup
     \mathbb{L}(\mathcal{A},>\lambda) $.
\end{definition}
~~
\begin{definition}
     The pair $ (\mathcal{A},\lambda) $ is \textit{equivalent under cutpoint
separation} to the pair
     $ (\mathcal{A^{'}},\lambda^{'}) $,
     denoted as $ (\mathcal{A},\lambda) \equiv
(\mathcal{A^{'}},\lambda^{'}) $,      if
     \[
             \begin{array}{lcl}
                     \mathbb{L}(\mathcal{A},< \lambda) & = &
\mathbb{L}(\mathcal{A^{'}},< \lambda^{'}) \\
                     \mathbb{L}(\mathcal{A},= \lambda) & = &
\mathbb{L}(\mathcal{A^{'}},= \lambda^{'}) \\
                     \mathbb{L}(\mathcal{A},> \lambda) & = &
\mathbb{L}(\mathcal{A^{'}},> \lambda^{'}), \\
             \end{array}
     \]
     where $ \mathcal{A} $, $ \mathcal{A^{'}} $ are automata and $
\lambda, \lambda^{'} \in \mathbb{R} $ are cutpoints.
\end{definition}
~~
\begin{definition}
     The \textit{language recognized by automaton $ \mathcal{A} $
with cutpoint} $
\lambda \in \mathbb{R} $ is defined as
     \[ \mathbb{L}(\mathcal{A},\lambda)=\mathbb{L}(\mathcal{A},>\lambda). \]
     $ \mathbb{L}(\mathcal{A},\lambda) $ is said to be \textit{recognized by
automaton $ \mathcal{A} $ with
     one-sided cutpoint} $ \lambda \in \mathbb{R} $ if $
\mathbb{L}(\mathcal{A},<\lambda)=\emptyset $.
\end{definition}
~~
\begin{definition} \cite{Pa71}
     \begin{enumerate}
              \item The languages recognized by PFA's with cutpoint $ \lambda \in [0,1) $ constitute the class of
                   \textit{stochastic languages} (S$ ^{>} $). 
                   The collection of languages whose complements are stochastic is the class co-S$ ^{>} $.
             \item Languages of the form $ \mathbb{L}(\mathcal{P},=\lambda) $, 
             for any PFA $ \mathcal{P} $, and any $ \lambda \in [0,1] $, constitute the class S$ ^{=} $.
             \item Languages of the form $ \mathbb{L}(\mathcal{P},\neq\lambda) $, for any PFA $ \mathcal{P} $,
             and any $ \lambda \in [0,1] $, constitute the class  
             \textit{exclusive stochastic languages} (S$ ^{\neq} $).
     \end{enumerate}
\end{definition}
~~
\begin{remark}
	\label{remark:cutpoint-unbounded}
	In the study of complexity classes defined in terms of Turing machines, ``recognition with cutpoint" 
	is used synonymously with ``unbounded-error recognition". 
	This usage does not provide an appropriate coverage of the intuitive concept of unbounded-error 
	computation that we described in Section \ref{section:Introduction} in the case of PFA's: 
	Given a PFA $ \mathcal{P} $ which recognizes a language $ L $ with cutpoint, 
	one can build a new PFA $ \mathcal{P}^{'} $ for the complement of $ L $ by just switching the statuses of 
	the accepting and non-accepting states of $ \mathcal{P} $. 
	Since $ \mathcal{P}^{'} $ accepts any member of $ \overline{L} $ with greater probability than any nonmember, 
	we say that it recognizes $ \overline{L} $ with unbounded error. 
	However, since S$ ^{>} $ is not known to be closed under complementation, we do not know in general 
	whether $ \overline{L} $ is stochastic or not. 
	For this reason, we take S$ ^{>} $ $ \cup $ co-S$ ^{>} $ to be the class of languages recognized 
	with unbounded error by PFA's
	\footnote{Note that S$ ^{>}_{rat} $, the class of languages recognized with cutpoint $ \frac{1}{2} $
	by PFA's whose transition matrices contain only rational numbers, is known \cite{Tu69B} 
	to be closed under complementation. It is however customary to define PFA's and QFA's with general transition
	probabilities/amplitudes, as we did in Definitions \ref{definition:1pfa} and \ref{definition:kwqfa}, 
	in the finite automata literature, and we follow this convention. 
	See Section \ref{section:ConcludingRemarks} for more on this point.}.
\end{remark}
~~
\begin{definition}~
     \begin{enumerate}
             \item  The languages recognized by KWQFA's (MCQFA's) with cutpoint $ \lambda  \in [0,1) $ constitute
             the class QL (MCL).
             \item The languages recognized by KWQFA's (MCQFA's) with cutpoint $ 0 $,
            i.e., those of the form $ \mathbb{L}(\mathcal{M},0) $,
for any KWQFA (MCQFA) $ \mathcal{M} $,
            constitute the class NQL (NMCL).
     \end{enumerate}
\end{definition}

As mentioned before, nondeterministic computation corresponds to
recognition with cutpoint $ 0 $, and so NMCL and NQL denote the classes
of languages recognized by nondeterministic MCQFA's and KWQFA's, respectively.
\\
\begin{fact}
     \label{fact:GPFA-GPFA}
     \cite{Tu69} Let $ \mathcal{G}_{1} $ be a GPFA  and $\lambda_{1} \in
\mathbb{R} $ be a cutpoint.
             For any cutpoint $ \lambda_{2} \in \mathbb{R} $, there
exists a GPFA
$ \mathcal{G}_{2} $  such that
             $ (\mathcal{G}_{1},\lambda_{1}) \equiv
(\mathcal{G}_{2},\lambda_{2}) $.
\end{fact}
~~
\begin{fact}
     \label{fact:PFA-PFA}
     \cite{Pa71} Let $ \mathcal{P}_{1} $ be a PFA  and $ \lambda_{1} \in
[0,1) $ be a cutpoint.
     For any cutpoint $ \lambda_{2} \in (0,1) $, there exists a PFA $
\mathcal{P}_{2} $  such that
             $ (\mathcal{P}_{1},\lambda_{1}) \equiv
(\mathcal{P}_{2},\lambda_{2}) $.
\end{fact}
~~
\begin{fact}
     \label{fact:PFA-KWQFA}
     \cite{YS09D} For any PFA $ \mathcal{P} $, there exists a KWQFA
$\mathcal{M} $ such that
             $ (\mathcal{P},\frac{1}{2}) \equiv (\mathcal{M},\frac{1}{2}) $.
\end{fact}
~~
\begin{fact}
     \label{fact:KWQFA-GPFA}
     \cite{LQ08,YS09D} For any KWQFA $ \mathcal{M} $ and cutpoint $
\lambda \in [0,1) $, there exists a
             GPFA $ \mathcal{G} $ such that $ (\mathcal{M},\lambda) \equiv
(\mathcal{G},\lambda) $.
\end{fact}
~~
\begin{fact}
     \label{fact:GPFA-PFA}
     \cite{Tu69} For any GPFA $ \mathcal{G} $ and cutpoint $ \lambda_{1}
\in \mathbb{R} $,
             there exist a PFA $ \mathcal{P} $ and a cutpoint $
\lambda_{2} \in
(0,1) $ such that
             $ (\mathcal{G},\lambda_{1}) \equiv (\mathcal{P},\lambda_{2}) $.
\end{fact}
~~
\begin{fact}
     \label{fact:MCLpropersubsetS}
     \cite{BC01B} MCL $ \subsetneq $ S$ ^{>} $.
\end{fact}
~~
\begin{fact}
     \label{fact:MCLpropersubsetQL}
     \cite{Br98} Any MCQFA with $ n $ states can be simulated by a KWQFA 
     with $ 2n $ states, so MCL $ \subseteq $ QL and NMCL $ \subseteq $ NQL.
\end{fact}
~~
\begin{fact}
     \label{fact:RLpropersubsetSLneqSLeq}
     \cite{Pa71} The class of regular languages is a proper subset of both
S$ ^{\neq} $ and S$ ^{=} $.
\end{fact}
~~
\begin{fact}
     \label{fact:SLneq-propersubsetSL}
     \cite{Pa71} S$ ^{\neq} $ $ \subsetneq $ S$ ^{>} $ and
     S$ ^{>} \setminus $ S$ ^{=} \neq \emptyset $.
\end{fact}
~~
\begin{fact}
       \label{fact:RLpropersubsetNQL}
       \cite{BP02,NIHK02}
       The class of regular languages is a proper subset of NQL.
\end{fact}
~~\\
By Facts \ref{fact:PFA-KWQFA}-\ref{fact:GPFA-PFA}, QL $ = $ S$ ^{>} $, 
PFA's and KWQFA's have the same language recognition power with general cutpoint and 
with unbounded error (Remark \ref{remark:cutpoint-unbounded}). 
It has in fact been shown \cite{YS09A} that all one-way QFA models \cite{Pa00,Ci01,BMP03,Hi08} 
that generalize the KWQFA are also equivalent to the PFA in this regard. 
(See Subsection \ref{subsection:more-general-QFAs} for more on this.)

We are interested in the case of one-sided unbounded error, where one
of the two responses that
the machine can output about the membership of the input string in the
recognized language is
correct with certainty, and the other response has a nonzero
probability of being correct.
We say that such an automaton has \textit{positive one-sided error}
if it rejects non-members of its language with certainty.
This corresponds to recognition with cutpoint 0.
The opposite case is called \textit{negative one-sided error},
where the language in question is of the form $
\mathbb{L}(\mathcal{A},\neq 1)$,
recalling that, when $\mathcal{A}$ is a PFA or a QFA,
$f_{\mathcal{A}}$ has range $ [0,1] $.

PFA's can recognize all and only the regular languages with cutpoint $
0 $ \cite{Ma93}.
KWQFA's can do more than that, as will be characterized in the next section.

\section{Languages Recognized with One-sided Error} \label{section:MainResults}
We start the presentation of our main result by stating a fact which
will be useful in several proofs in the paper.
\\
\begin{lemma}
    \label{lemma:Sneq-GPFAonesidedcutpoint0}
   For any language $ L $, $ L \in $ S$ ^{\neq} $ if and only if
there exists a GPFA that recognizes $ L $
   with one-sided cutpoint $ 0 $.
\end{lemma}
\begin{proof}
    The forward direction is proven on page 171 of \cite{Pa71}.
    In the reverse direction,
    if a GPFA recognizes $ L $ with one-sided cutpoint $ 0 $,
    then  $ L \in $ S$^{\neq} $ by Fact \ref{fact:GPFA-PFA}.
\end{proof}

\subsection{A characterization of NQL}
~~

\begin{lemma}
   \label{lemma:S-QL_0}
   S$ ^{\neq} \subseteq $ NQL.
\end{lemma}
\begin{proof}
   If $ L \in S^{\neq}$, then there exists an $ n $-state  PFA
   $ \mathcal{P} = (S,\Sigma,\{\mathsf{A}_{\sigma \in \Gamma} \},F)
$ such that
           $ L=\mathbb{L}(\mathcal{P},\neq\frac{1}{2}) $.
   We define $ S^{'} $, $ \mathsf{v}_{0}^{'} $, and
   $ \{ \mathsf{A}_{\sigma \in \Gamma}^{'} \} $ as follows:
   \begin{enumerate}
           \item $ S^{'} = S \cup \{ s_{n+1}, s_{n+2}, s_{n+3} \} $;
           \item $ \mathsf{v}_{0}^{'}=(1,0,\cdots,0) $ is a $ 1
\times (n+3) $-dimensional row vector;
           \item Each $ \mathsf{A}_{\sigma}^{'} $ is a $ (n+3)
\times (n+3) $-dimensional matrix:
           \[
           \mathsf{A}_{\cent}^{'}= \left(
                   \begin{array}{c|ccc}
                           \frac{1}{2}\mathsf{A}_{\cent}[r_{1}] & 0
& 0 & \frac{1}{2} \\
                           \hline
                           1~0~\cdots~0 & 0 & 0 & 0 \\
                           \vdots & & \vdots \\
                           1~0~\cdots~0 & 0 & 0 & 0 \\
                   \end{array}
           \right), ~~~~
           \mathsf{A}_{\sigma \in \Sigma}^{'}= \left(
                   \begin{array}{c|ccc}
                   \\
                           \mathsf{A}_{\sigma} &  & 0_{n \times 3} &  \\\\
                           \hline
                            & 1 & 0 & 0 \\
                           0_{3 \times n}   & 0 & 1 & 0 \\
                            & 0 & 0 & 1 \\
                   \end{array}
           \right),
           \]
           \[
           \mathsf{A}_{\dollar}^{'}=
           \left(
                   \begin{array}{c|ccc}
                   \\
                           \mathsf{A}_{\dollar} &  & 0_{n \times 3} &  \\\\
                           \hline
                            & 1 & 0 & 0 \\
                           0_{3 \times n}   & 0 & 1 & 0 \\
                            & 0 & 0 & 1 \\
                   \end{array}
           \right)
           \left(
                   \begin{array}{c|ccc}
                           & t_{1,1} & t_{1,2} & 0 \\
                           0_{n \times n} & & \vdots & \\
                           & t_{n,1} & t_{n,2} & 0 \\
                           \hline
                            & 1 & 0 & 0 \\
                           0_{3 \times n}   & 0 & 1 & 0 \\
                            & \frac{-1}{2} & \frac{1}{2} & 0 \\
                   \end{array}
           \right),
            \]
   where $ \mathsf{A}_{\cent}[r_{1}] $ is the first row of $
\mathsf{A}_{\cent} $;
           $ t_{i,1}=1 $ and $ t_{i,2}=0 $ when $ s_{i} \in F $, and
           $ t_{i,1}=0 $ and $ t_{i,2}=1 $ when $ s_{i} \notin F $
           for $ 1 \le i \le n $.
   \end{enumerate}
   For a given input $ w \in \Sigma^{*} $, $ \mathsf{w}=\cent w
\dollar $, let
   $ \mathsf{v}_{|\mathsf{w}|}^{'} = \mathsf{v}_{0}^{'}
\mathsf{A}_{\cent}^{'} \mathsf{A}_{w_{1}}^{'} \cdots
   \mathsf{A}_{w_{|w|}}^{'} \mathsf{A}_{\dollar}^{'} $.
   It is easily verified that this computation ``imitates"
   the processing of $w$ by $\mathcal{P}$; the first $n$ entries of
   the manipulated vector $\mathsf{v}^{'}$  contain exactly the
state vector of $\mathcal{P}$
   (multiplied by $ \frac{1}{2} $) in the corresponding steps of
its execution. The last matrix multiplication results in
   \[ \mathsf{v}_{|\mathsf{w}|}^{'}=\left(0_{1 \times n} ~\vline~
   \frac{2f_{\mathcal{P}}(w)-1}{4},\frac{3-2f_{\mathcal{P}}(w)}{4},0
\right).\]
   The $ (n+1)^{st} $ entry of $ \mathsf{v}_{|\mathsf{w}|}^{'} $
equals $ 0 $ if and only if $ w \notin L $.

   Using a modified version of the PFA simulation method described
in \cite{YS09D}, we can construct a
   KWQFA $ \mathcal{M}=(Q, \Sigma, \{\mathsf{U}_{\sigma \in \Gamma} \},
   Q_{acc}, Q_{rej} ) $ recognizing $ L $ with cutpoint 0:
   For each $ \sigma \in \Gamma $, $ \mathsf{U}_{\sigma} $ is built
according to the template
   \[
           \mathsf{U}_{\sigma}^{\mathtt{T}}=\left(
                   \begin{array}{c|c|c}
                           &
                           \\
                           ~\mathsf{c}_{\sigma}\mathsf{A}_{\sigma}^{'}~ &
                           ~\mathsf{c}_{\sigma}B_{\sigma}~ &
                           ~\mathsf{c}_{\sigma}C_{\sigma}~
                            \\
                           &  \\
                           \hline
                           \multicolumn{3}{c}{} \\
                           \multicolumn{3}{c}{D_{\sigma}} \\
                           \multicolumn{3}{c}{}
                   \end{array}
           \right),
    \]
    where $ B_{\sigma}=[b_{i,j}] $ is a lower triangular matrix, and
    $ C_{\sigma}=[c_{i,j}] $ is a diagonal matrix. The entries of $
\mathsf{U}_{\sigma} $ are computed iteratively using the following
procedure:
     \begin{enumerate}
             \item The entries of $ B_{\sigma} $ and $ C_{\sigma} $
are set to 0.
             \item The entries of $ B_{\sigma} $ are updated to make
the rows of
                     $ \left( \mathsf{A}_{\sigma}^{'} \mid
B_{\sigma} \right) $ pairwise
orthogonal.
                     Specifically, \\\\
                     \begin{tabular}{ll}
                             ~~~~ & for $ i=1 \cdots n+2 $ \\
                                      & ~~~~set $ b_{i,i}=1 $ \\
                                      & ~~~~for $ j=i+1 \cdots n+3 $ \\
                                      & ~~~~~~~~set $ b_{j,i} $ to
some value so that the $ i $th and $ j $th
                                             rows become orthogonal \\
                                      & set $l_{max} $ to the
maximum of the lengths (norms) of the rows of
                                                             $
\left( \mathsf{A}_{\sigma}^{'} \mid B_{\sigma} \right) $
                     \end{tabular}
                     \\\\
             \item The diagonal entries of $ C_{\sigma} $ are
updated to make the
length of each row of
                     $ \left( \mathsf{A}_{\sigma}^{'} \mid
B_{\sigma} \mid C_{\sigma}
\right) $ equal to $ l_{max} $.
                     Specifically,\\\\
                     \begin{tabular}{ll}
                              ~~~~ & for $ i=1 \cdots n+3 $ \\
                                       & ~~~~set $ l_{i} $ to the
current length of the $ i $th row of
                                       $ \left(
\mathsf{A}_{\sigma}^{'} \mid B_{\sigma} \right) $
                                        \\
                                     & ~~~~set $ c_{i,i} $ to $
\sqrt{l_{max}^{2}-l_{i}^{2}} $
                     \end{tabular}
                     \\\\
             \item Set $ c_{\sigma} $ to $ \dfrac{1}{l_{max}} $.
             \item The entries of $ D_{\sigma} $ are selected to make $
\mathsf{U}_{\sigma}^{\mathtt{T}} $
                     a unitary matrix. The transpose accounts for
the difference
                     between the probabilistic and quantum vector notations.
   \end{enumerate}
   The state set $ Q = Q_{non} \cup Q_{acc} \cup Q_{rej}  $ is specified as:
   \begin{enumerate}
           \item $ q_{n+1} \in Q_{acc} $ corresponds to state $ s_{n+1} $;
           \item $ q_{n+2} \in Q_{rej} $ corresponds to state $ s_{n+2} $;
           \item $ \{q_{1},\cdots,q_{n},q_{n+3}\} \in Q_{non} $
correspond to the remaining states of
           $ S^{'} $, where $ q_{1} $ is the start state;
           \item All the new states that are defined during the
construction of
           $ \{\mathsf{U}_{\sigma \in \Gamma}\} $ are rejecting ones.
   \end{enumerate}
   $ \mathcal{M} $ simulates the computation of $ \mathcal{P} $
   for a given input string $ w \in \Sigma^{*} $ by
representing the probability of each state $ s_{j}$ by the amplitude
of the corresponding state $ q_{j}$; specifically,
this amplitude equals
$\mathsf{c}_{\cent} \left( \prod_{i=1}^{k} \mathsf{c}_{w_{i}} \right)
\times \mathsf{v}^{'}(j) $
immediately after the $(k+1)^{st}$ step of the computation
\cite{YS09D}, where $ k \le |w| $.
   The transitions from the $ 2n+6 $ states added
     during the construction of $ \mathsf{U}_{\sigma \in \Gamma } $ for
ensuring unitarity do not interfere
     with this simulation, since the computation halts immediately on the
     ``branches" where these states are entered.
   Therefore, the top $ n+3$ entries of the state vector of $
\mathcal{M} $ equal
   \[ \mathsf{c}_{\cent} \left( \prod_{i=1}^{|w|}
\mathsf{c}_{w_{i}} \right) \mathsf{c}_{\dollar}
           \left(0_{1 \times n} ~\vline~\frac{2f_{\mathcal{P}}(w)-1}{4},
           \frac{3-2f_{\mathcal{P}}(w)}{4},0 \right)^{\mathtt{T}}
    \]
just before the last measurement on the right end-marker.
Since the amplitude of the only accepting state is nonzero if and only
if $ w \in L $,
$ L $ is recognized by $ \mathcal{M} $ with cutpoint $ 0 $.
\end{proof}
~~
\begin{lemma}
   NQL$ \subseteq  $ S$ ^{\neq} $.
\end{lemma}
\begin{proof}
   By Fact \ref{fact:KWQFA-GPFA}, there exists a GPFA with one-sided cutpoint 0
   for any member $L$ of NQL. By Lemma \ref{lemma:Sneq-GPFAonesidedcutpoint0},
   $L$ is an exclusive stochastic language.
\end{proof}
~~
\begin{theorem}
   \label{theorem:maintheorem}
   S$ ^{\neq} $ = NQL.
\end{theorem}
~~
\begin{corollary}
      \label{corollary:SLeq-equal-coNQL}
   S$ ^{=} $ is precisely the class of languages that can be recognized
   with negative one-sided error by KWQFA's.
\end{corollary}
~~\\
The superiority of KWQFA's over PFA's in the one-sided error setting
now follows from
Fact \ref{fact:RLpropersubsetSLneqSLeq}. By Fact
\ref{fact:SLneq-propersubsetSL}, there exist
languages that KWQFA's can recognize with two-sided, but not one-sided error.
The class of these languages is precisely 
$ ( \mbox{S} ^{>} \cup \mbox{co-S} ^{>} ) \setminus ( \mbox{S} ^{\neq} \cup \mbox{S} ^{=} ) $ (Remark 2.1).
Note that the above results also establish that the class of languages
recognized by NQFA's is not closed under complementation (Fact \ref{fact:closure-properties-Seq-Sneq}).
\subsection{More general QFA models} \label{subsection:more-general-QFAs}
Several one-way QFA models (like \cite{Na99,Pa00,Ci01,BMP03,Hi08}, and
the one-way
version of the machines of \cite{AW02},) that generalize the KWQFA have
appeared in the literature. In the bounded-error case, some of these
generalized machines recognize more languages than the KWQFA. We claim
that the classes of languages recognized by the nondeterministic
versions of all these automata are identical to each other, and
they coincide with NQL.

We demonstrate this fact for one of the most general models, namely,
the quantum finite automaton with ancilla qubits (QFA-A) \cite{Pa00},
which can simulate all known one-way QFA models.
Let us give the name QFA-A$ _{0} $ to the class of
languages recognized with cutpoint 0 by QFA-A's.
For any QFA-A $ \mathcal{M} $, there exists a GPFA that computes
exactly the same
acceptance probability function as $ \mathcal{M} $ \cite{YS09G},  so
QFA-A$ _{0}  \subseteq $ S$ ^{\neq} $ by Lemma
\ref{lemma:Sneq-GPFAonesidedcutpoint0}.
Since any KWQFA can be simulated by a QFA-A, NQL $ \subseteq $
QFA-A$ _{0} $.
Therefore, QFA-A$ _{0} = $ NQL$ = $ S$ ^{\neq} $.

\subsection{Space efficiency of QFA's with cutpoint 0}
It is well known \cite{AF98,MPP01} that some infinite families of languages can
be recognized with one-sided bounded error by just tuning the
transition amplitudes of a QFA
with a constant number of states,
whereas the sizes of the corresponding PFA's grow without bound.
After a simple example, we will argue that this advantage
is also valid in the unbounded error case.
\\
\begin{definition} For $ m \in \mathbb{Z}^{+} $, $ L_{m} \subseteq
\{a\}^{*} $ is defined as
    \[ L_{m}=\{a^{i} \mid i \mod(m) \neq 0 \}. \]
\end{definition}
~~
\begin{theorem}
     \label{theorem:L_m}
     For $ m>1 $, $ L_{m} $ can be recognized by a 2-state
MCQFA\footnote{There is an equivalent 4-state KWQFA.}
     with cutpoint $ 0 $.
\end{theorem}
\begin{proof}
    $ \mathcal{M} $ begins the computation at state $ q_{0} $, and
each transition with the symbol $a$
    corresponds to a rotation\footnote{For details of a similar
construction for a nonregular language,
    see \cite{BC01B}.} by angle $ \frac{\pi}{m} $ in the $
\ket{q_{0}} $-$ \ket{q_{1}} $ plane,
    where $ q_{1} $ is the accepting state.
\end{proof}

For any positive $n$, it is known \cite{Ma93} that every $n$-state PFA
with cutpoint $ 0 $ has an equivalent
nondeterministic finite automaton with the same number of states.
Therefore, only finitely many distinct languages can be recognized
with one-sided
unbounded error by PFA's with at most $n$ states.

Combining this with the fact that any $n$-state PFA with cutpoint $ 0
$ can be simulated by a
KWQFA with $ 2n+4 $ states using a
simple adaptation of the technique of \cite{YS09D},
the superiority of QFA's over PFA's in this regard is
established\footnote{Note that a QFA-A can realize this simulation
with just $n$ states.}.

\section{S$ ^{\neq} $, S$ ^{=} $, and Languages Recognized with Two-sided Error}
To gain a better understanding of the classes of languages
recognizable by positive one-sided, negative one-sided, and
necessarily two-sided error by QFA's, we examine some examples from
each of those families. Bertoni and Carpentieri \cite{BC01B} showed
that $ L_{neq} $ is in NMCL, and that its complement, say, $
L_{eq}=\{w\in\{a,b\}^{*} \mid |w|_{a}
=|w|_{b}\} $,
is not in MCL. Now that we have Theorem \ref{theorem:maintheorem}, we
can use the well-known
results \cite{Pa71, Ma93} from the PFA literature that state that
$ L_{eq} \in $  S$ ^{=} $, $ L_{neq} \in $ S$ ^{\neq} $, but not
vice versa, to conclude that stronger QFA variants also can not recognize $ L_{eq} $ with positive
one-sided error, and neither can they recognize $ L_{neq} $ with
negative one-sided error. Similarly, L\={a}ce \textit{et al.} \cite{LSF09}
proved recently that the complement of the palindrome language $
L_{pal}=\{w \in \{a,b\}^{*}
\mid w=w^{r} \} $ is in NQL. We can show the corresponding
result for $ L_{pal} $ using the following fact:
\\
\begin{fact} \cite{Di71}
   \label{fact:dieu}
   Let $ L \in $ S$ ^{=} $. Then there exists a natural number $ n
\ge 1 $, such that for any strings
   $ u, v, y \in \Sigma^{*} $,
   \[ \mbox{if } uv, uyv, \cdots, uy^{n-1}v \in L, \mbox{then }
uy^{*}v \subseteq L. \]
\end{fact}
~~
\begin{theorem}
      \label{theorem:L_pal}
    $ L_{pal} \notin $  S$ ^{\neq} $.
\end{theorem}
\begin{proof}
    Suppose that $ L_{pal} \in $ S$ ^{\neq} $. Then $
\overline{L_{pal}} \in $ S$ ^{=} $.
    Let $ u=a^{n}b $, $ y=a $, and $ v=\varepsilon $.
    \[ a^{n}b, a^{n}ba,\cdots,a^{n}ba^{n-1} \in \overline{L_{pal}} \]
    imply that $ a^{n}ba^{n} \in $ $ \overline{L_{pal}} $ by Fact
\ref{fact:dieu}.
    Since this string is actually a member of $ L_{pal} $,  we have
a contradiction.
\end{proof}
We will now exhibit some languages which can only be recognized by
two-sided error by a QFA.
\\
\begin{theorem}
      \label{theorem:L_twosided-1}
    $ L=\{aw_{1} \cup bw_{2} \mid w_{1} \in L_{eq}, w_{2} \in
              L_{neq} \}  \in $ S$ ^{>} \setminus ($S$ ^{=} \cup $
S$ ^{\neq} )$.
\end{theorem}
\begin{proof}
    Suppose that $ L \in $ S$ ^{\neq} $, then there exists a GPFA
    \[ \mathcal{G}=(S,\Sigma,\{\mathsf{A}_{\sigma \in \{a,b\}}
\},\mathsf{v}_{0},\mathsf{f}) \]
    recognizing $ L$ with one-sided cutpoint $ 0 $. The GPFA
    \[ \mathcal{G}^{'}=(S,\Sigma,\{\mathsf{A}_{\sigma \in \{a,b\}}
\},\mathsf{v}_{0}\mathsf{A}_{a},\mathsf{f}) \]
    recognizes $ L_{eq} $ with one-sided cutpoint 0, meaning that $
L_{eq} \in $ S$ ^{\neq} $.
    This contradicts the well-known fact mentioned in the first
paragraph of this section.
    Suppose now that $ L \in $ S$ ^{=} $,  then
    \[ \overline{L}=\{\varepsilon \cup aw_{2} \cup bw_{1} \mid w_{1}
\in L_{eq}, w_{2} \in L_{neq} \} \]
    is in S$ ^{\neq} $, which also results in a contradiction for
the same reason.
    Since both $ L_{eq} $ and its complement are stochastic, it is
not difficult to show that $ L $ is stochastic.
\end{proof}
~~
\begin{lemma}
      \label{lemma:L_twosided-2}
      $ L_{lt}= \{w \in \{a,b\}^{*} \mid |w|_{a} < |w|_{b} \} $ $ \notin
($S$ ^{=} \cup $  S$ ^{\neq} )$.
\end{lemma}
\begin{proof}
      Suppose that $ L_{lt} \in $ S$ ^{=} $. Let $ u=\varepsilon  $, $ y=a
$, and $ v=b^{n} $.
      \[ b^{n}, ab^{n},\cdots,a^{n-1}b^{n} \in L_{lt} \]
      imply that $ a^{n}b^{n} \in $ $ L_{lt} $ by Fact \ref{fact:dieu}.
      Since this string is actually a member of $ \overline{L_{lt}} $,  we
have a contradiction.

      Similarly, suppose that $ L_{lt} \in $ S$ ^{\neq} $, or $
\overline{L_{lt}} \in $ S$ ^{=} $.
      Let $ u=a^{n} $, $ y=b $, and $ v=b $.
      \[ a^{n}b, a^{n}b^{2},\cdots,a^{n}b^{n} \in \overline{L_{lt}} \]
      imply that $ a^{n}b^{n+1} \in \overline{L_{lt}} $ by Fact
\ref{fact:dieu}.
      Since this string is actually a member of $ L_{lt} $,  we have
a contradiction.
\end{proof}
~~
\begin{corollary}
      \label{corollary:L_twosided-2}
      $ L_{lt} \in $ S$ ^{>} \setminus ($S$ ^{=} \cup $  S$ ^{\neq} )$.
\end{corollary}
\begin{proof}
      This follows from Lemma \ref{lemma:L_twosided-2} and the fact that $
L_{lt} \in  $ S$ ^{>} $ \cite{Ra92,Ka89}.
\end{proof}
~~
\begin{theorem}
      \label{theorem:L_twosided-3}
      $ L_{eq \cdot b}=
L_{eq} \cdot b^{+} $ $ \in $
      S$ ^{>} \setminus ($S$ ^{=} \cup $  S$ ^{\neq} )$.
\end{theorem}
\begin{proof}
      The proof of $ L_{eq \cdot b} \notin ($S$ ^{=} \cup $  S$
^{\neq} ) $ uses the
      setup presented in Lemma \ref{lemma:L_twosided-2}, i.e.,
      \begin{enumerate}
              \item select $ u=\varepsilon  $, $ y=a $, and $ v=b^{n} $ to
contradict with $ L_{eq \cdot b} \in $ S$ ^{=} $,
              \item select $ u=a^{n} $, $ y=b $, and $ v=b $ to
contradict with $
L_{eq \cdot b} \in $ S$ ^{\neq} $.
      \end{enumerate}
      Any string $ w $ is a member of $ L_{eq \cdot b} $ if and only if it
      has the following three properties:
      \begin{itemize}
              \item $ w $ ends with $ b $.
              \item $ w \in L_{lt} $.
              \item Let $ u $ be the longest prefix of $ w $ ending with $ a $
              ($ u = \varepsilon $ if $ w \in \{b^{*}\} $). Then, $ u \in
\overline{L_{lt}} $.
      \end{itemize}
      Since these properties can be checked easily by a two-way PFA with
bounded error,
      $ L_{eq \cdot b} \in $ S$ ^{>} $ \cite{Ra92,Ka89}.
\end{proof}

We conclude this section by showing the stochasticity of an important
family of languages.
\\
\begin{definition}
      \label{definition:L_word-problem}
      \cite{LZ77}
      The \textit{word problem} for a group is the problem of deciding
whether or not a product of group elements
      is equal to the identity element.
\end{definition}
~~
\begin{definition}
      \label{definition:L_word-problem-free-group}
      Let $ \mathcal{G}^{k}=(G,\circ) $ be a finitely generated free
group with a basis
      \[ \Sigma
=\{\sigma_{1},\ldots,\sigma_{k},\sigma_{1}^{-1},\ldots,\sigma_{k}^{-1}\},
\]
      where $ k \in \mathbb{Z}^{+} $ is the rank of $ \mathcal{G}^{k} $.
      $ L_{wp}(\mathcal{G}^{k}) \subseteq \Sigma^{*} $ is the
language defined as
      \[ L_{wp}(\mathcal{G}^{k})=\{ w=w_{1} \cdots w_{|w|} \mid
      w_{i} \in \Sigma, 1 \le i \le |w|,
      w_{1} \circ \cdots \circ w_{|w|} = \imath \}, \]
      where $ \imath \in G $ is the identity element of $ \mathcal{G}^{k} $.
\end{definition}
~~
\begin{fact}
      \label{fact:isomorphic-free-group}
      (Page 1 of \cite{LS77})
      Let $ \mathcal{G}_{1}^{k_{1}} $ and $ \mathcal{G}_{2}^{k_{2}} $ be
finitely generated free groups.
      Then $ \mathcal{G}_{1}^{k_{1}} $ and $ \mathcal{G}_{2}^{k_{2}} $ are
isomorphic if and only if $ k_{1} = k_{2} $.
\end{fact}
~~
\begin{corollary}
      \label{corollary:isomorphic-free-group}
      $ L_{wp}(\mathcal{G}_{1}^{k_{1}}) $ and $
L_{wp}(\mathcal{G}_{2}^{k_{2}}) $ are isomorphic
      if and only if $ k_{1} = k_{2} $,
      where $ \mathcal{G}_{1}^{k_{1}} $ and $ \mathcal{G}_{2}^{k_{2}} $ are
finitely generated free groups.
\end{corollary}
~~\\
As a generic name, $ L_{wp}(k) $ can be used instead of $ L_{wp}(\mathcal{G}^{k}) $
due to Corollary \ref{corollary:isomorphic-free-group}, where $ k \in \mathbb{Z}^{+} $.
\\
\begin{fact}
      \label{fact:word-problem-one-SL}
      \cite{Tu81}
      $ L_{wp}(1) \in $ S$ ^{>} $.
\end{fact}
~~
\begin{fact}
      \label{fact:free-group-co-NMCL}
      \cite{BP02}
      $ L_{wp}(k) \in $ co-NMCL, the class of languages whose
complements are in NMCL,
       for any $ k \in \mathbb{Z}^{+} $.
\end{fact}
~~
\begin{corollary}
      \label{corollary:free-group-SLeq}
      $ L_{wp}(k) \in $ S$ ^{=} $ for any $ k \in \mathbb{Z}^{+} $.
\end{corollary}
~~\\
We will now provide a proof of the following theorem.
\\
\begin{theorem}
      \label{theorem:word-problem-is-stochastic}
      $ L_{wp}(k) \in $ S$ ^{>} $ for any $ k \ge 2 $.
\end{theorem}
~~

In fact, Theorem \ref{theorem:word-problem-is-stochastic} was stated
as a corollary on page 1463 of \cite{BP02}, but the purported proof
there
was based on the claim that co-NMCL$ \subseteq $MCL$ \subseteq $ S$ ^{>} $.
It is however known \cite{BC01B}, as we mentioned above, that a member
of co-NMCL ($ L_{eq} $) lies outside MCL.
Furthermore, the same demonstration can be easily extended to $
L_{wp}(k) $, where $ k \in \mathbb{Z}^{+} $.
\\
\begin{corollary}
      \label{corollary:word-problem-is-not-in-MCL}
      $ L_{wp}(k) \notin $ MCL for any $ k \in \mathbb{Z}^{+} $.
\end{corollary}
~~

Since it is still an open problem whether S$ ^{=} \subseteq $ S$ ^{>} $ or not,
we cannot use Corollary \ref{corollary:free-group-SLeq} directly to
prove Theorem \ref{theorem:word-problem-is-stochastic}.
Instead, we will focus on a subclass of S$ ^{=} $ that is known to be
a subset of S$ ^{>} $.
\\
\begin{definition}
      \label{definition:SLeq-rat}
      \cite{Tu69B}
      S$ ^{=}_{rat} $ is the class of the languages of the form $
\mathbb{L}(\mathcal{G},=\lambda) $,
      where $ \mathcal{G} $ is a \textit{rational} GPFA, (i.e. one whose
transition matrices and initial and final vectors contain only
rational numbers,)
      and $ \lambda $ is a rational number. Additionally, S$ ^{\neq}_{rat}
$ is the class of languages whose complements are in S$ ^{=}_{rat} $.
\end{definition}
~~
\begin{fact}
      \label{fact:SLeq-rat-subset-of-SL}
      \cite{Tu69B}
      S$ ^{=}_{rat} $ $ \subsetneq $ S$ ^{>} $.
\end{fact}
~~
\begin{definition}
      \label{definition:SQ3-rational}
      SO$ _{3}(\mathbb{Q}) $ is the group of rotations on $ \mathbb{R}^{3}
$ that are
      $ 3 \times 3 $ dimensional orthogonal matrices having only rational
entries with determinant $ +1 $.
\end{definition}
~~
\begin{fact}
       \label{fact:SQ3-rational-free-group}
       \cite{Sw94,DD06} For any $ k \ge 2 $, SO$ _{3}(\mathbb{Q}) $
       contains a free subgroup with rank $ k $, namely $ \mathcal{S}^{k} $.
\end{fact}
\begin{proof}[Proof of Theorem \ref{theorem:word-problem-is-stochastic}]
      For $ L_{wp}(k) $, we define a rational GPFA
      $ \mathcal{G}_{k}=(\{s_{1},s_{2},s_{3}\}, \Sigma,
      \{\mathsf{A}_{\sigma \in \Sigma}\},\mathsf{v}_{0},\mathsf{f}) $,
      where
      \begin{enumerate}
              \item $
\Sigma=\{R_{1},\ldots,R_{k},R_{1}^{-1},\ldots,R_{k}^{-1}\} $
                      is a basis of $ \mathcal{S}^{k} $;
              \item $ \mathsf{A}_{\sigma}=\sigma $ for each $ \sigma
\in \Sigma $;
              \item $ \mathsf{v}_{0}=(1,0,0) $;
              \item $ \mathsf{f}=(1,0,0)^{\mathtt{T}} $.
      \end{enumerate}
      It is obvious that $ w \in L_{wp}(k) $ if and only if
      $ \mathsf{A}_{1} \ldots \mathsf{A}_{w} = I_{3 \times 3} $ if and only if
      $ f_{\mathcal{G}_{k}}(w)=\mathsf{v}_{0} \mathsf{A}_{1} \ldots
\mathsf{A}_{w} \mathsf{f} = 1 $,
      where $ w \in \Sigma^{*} $.
      Thus, $ L_{wp}(k) \in $ S$ ^{=}_{rat} $ by selecting the
cutpoint as $ 1 $.
      We can conclude with Fact \ref{fact:SLeq-rat-subset-of-SL}.
\end{proof}

%
\section{Closure Properties}
%

The previously discovered closure properties of $ S^{>} $, $ S^{\neq} $ and $ S^{=} $ are listed below.
\\
\begin{fact}
      \label{fact:closure-properties-SL}~
      \begin{enumerate}
      \item $ S^{>} $ is not closed under union and intersection \cite{Fl73A,Fl73B,La74}.
      \item $ S^{>} $ is closed under union and intersection with a regular language \cite{Bu68,Tu68}.
              \item $ S^{>} $ is closed under reversal \cite{Tu69}.
              \item $ S^{>} $ is not closed under concatenation, Kleene closure, and homomorphism \cite{Co69,Tu71}.
              \item $ S^{>} $ is closed under complementation over unary alphabets \cite{Fl73A}.
      \end{enumerate}
\end{fact}
~~
\begin{fact}
      \label{fact:closure-properties-Seq-Sneq}~
      \begin{enumerate}
              \item Both S$ ^{\neq} $ and S$ ^{=} $ are closed under union and
intersection \cite{Pa71}.
              \item Neither S$ ^{\neq} $ nor S$ ^{=} $ is closed under
complementation \cite{Di71}.
              \item S$ ^{>} $  is closed under intersection with a
member of S$ ^{\neq} $ \cite{Pa71}.
      \end{enumerate}
\end{fact}

We will prove several new nontrivial closure properties of the
``one-sided" classes $ S^{\neq} $ and $ S^{=} $.

\subsection{Dissimilar closure properties of S$ ^{\neq} $ and S$ ^{=} $}

The proofs of the next few theorems use the capability of GPFA's to
implement nondeterministic branching by just adding the transition
matrices of the branches, and the nice
properties of computation with one-sided cutpoint 0.
\\
\begin{theorem}
     \label{theorem:one-sided-closure-concatenation} S$ ^{\neq} $ is
closed under concatenation.
\end{theorem}
\begin{proof}
     If $ L_{1} $, $ L_{2} \in $ S$ ^{\neq} $, then there exist two GPFA's
     $ \mathcal{G}_{1}= (S^{'},\Sigma,\{\mathsf{A}_{\sigma \in
\Sigma}^{'} \},\mathsf{v}_{0}^{'},\mathsf{f}^{'}) $
     and $ \mathcal{G}_{2}= (S^{''},\Sigma,\{\mathsf{A}_{\sigma \in
\Sigma}^{''}
     \},\mathsf{v}_{0}^{''},\mathsf{f}^{''}) $ such that $ L_{1} $
and $ L_{2} $ are recognized with one-sided
     cutpoint $ 0 $ by $ \mathcal{G}_{1} $ and $ \mathcal{G}_{2} $,
respectively.
     Let $ n_{1} $ and $ n_{2} $ be the sizes of the state sets $
S_{1} $ and $ S_{2} $, respectively.

     We construct a new GPFA $
\mathcal{G}=(S,\Sigma,\{\mathsf{A}_{\sigma \in \Sigma}
\},\mathsf{v}_{0},\mathsf{f}) $
     recognizing $ L=L_{1}L_{2} $ (concatenation of $ L_{1} $ and $
L_{2} $) with one-sided cutpoint $ 0 $.
     The details of $ \mathcal{G} $ are as follows:
     \begin{enumerate}
              \item The size of $ S $ is $ n= n_{1}+n_{2} $;
              \item $ \mathsf{v}_{0} $ is a $ 1 \times n $ row vector,
                      \begin{enumerate}
                              \item $ \mathsf{v}_{0}=(\mathsf{v}_{0}^{'} \mid
\mathsf{v}_{0}^{''}) $ if $ \varepsilon $ (empty string)
                                      belongs to $ L_{1} $, and
                              \item $
\mathsf{v}_{0}=(\mathsf{v}_{0}^{'} \mid 0_{1 \times
n_{2}}) $ if $ \varepsilon \notin L_{1} $;
                      \end{enumerate}
              \item $ \mathsf{f} $ is a $ n \times 1 $ column vector,
                      \begin{enumerate}
                              \item $ \mathsf{f}=((\mathsf{f}^{'})^{\mathtt{T}}
\mid(\mathsf{f}^{''})^{\mathtt{T}})^{\mathtt{T}} $
                                      if $ \varepsilon \in L_{2} $, and
                              \item $ \mathsf{f}=(0_{1 \times n_{1}} \mid
(\mathsf{f}^{''})^{\mathtt{T}})^{\mathtt{T}} $
                                      if $ \varepsilon \notin L_{2} $;
                      \end{enumerate}
              \item $ \{ \mathsf{A}_{\sigma \in \Sigma} \} $ is the set of $ n
\times n $ matrices,
\begin{equation}
      \label{equation:A_sigma}
       \mathsf{A}_{\sigma}=
                              \left(
                              \begin{array}{c|c}
                                      \mathsf{A}_{\sigma}^{'} &
~~\mathsf{X}_{\sigma}~~ \\ \hline
                                      0_{n_{2} \times n_{1}} &
\mathsf{A}_{\sigma}^{''}
                              \end{array}
                              \right),
\end{equation}
              where $ \mathsf{X}_{\sigma} $ is an $ n_{1} \times
n_{2} $ matrix, defined as
\begin{equation}
      \label{equation:X_sigma}
      \mathsf{X}_{\sigma}=
              \left(
                      \begin{array}{c|c|c|c}
                              & & ~~~~~~~~ &  \\
                              \mathsf{v}_{0}^{''}(1)
\mathsf{A}_{\sigma}^{'}\mathsf{f}^{'} ~&~
                              \mathsf{v}_{0}^{''}(2)
\mathsf{A}_{\sigma}^{'}\mathsf{f}^{'} ~&~
                              \cdots &~
                              \mathsf{v}_{0}^{''}(n_{2})
\mathsf{A}_{\sigma}^{'}\mathsf{f}^{'} \\
                              \underbrace{~~~~~~~~~~~~}_{\mbox{column 1}} &
\underbrace{~~~~~~~~~~~~}_{\mbox{column 2}} &  &
                              \underbrace{~~~~~~~~~~~~}_{\mbox{column
$ n_{2} $}}
                      \end{array}
              \right).
\end{equation}
     \end{enumerate}
     The idea behind the construction is that for a given input
string $ w \in \Sigma^{*} $,
     each prefix of $ w $, say, $ u \in \Sigma^{*} $ ($ w=uv $), is
checked for belonging to $ L_{1} $,
     and if so, the rest, $ v $, is checked for belonging to $ L_{2} $.

$ \mathcal{G} $ simulates $\mathcal{G}_{1}$ in the first $ n_{1} $
positions of its state vector. If $ \mathcal{G}_{1} $ accepts an input
prefix $u$ ending with $ \sigma $, the result of the multiplication
between that $ 1 \times n_{1}  $ row vector describing the distribution after
processing the first $ |u|-1 $ input symbols
and the column vector
$\mathsf{A}_{\sigma}^{'}\mathsf{f}^{'} $, that is,
$f_{\mathcal{G}_{1}}(u)$, will be positive. Otherwise,
$f_{\mathcal{G}_{1}}(u)=0$. By Equations \ref{equation:A_sigma} and
\ref{equation:X_sigma}, the
vector $f_{\mathcal{G}_{1}}(u) . \mathsf{v}_{0}^{''} $ will be
added to the last $ n_{2} $ positions of $ \mathcal{G} $'s state vector,
meaning that $\mathcal{G}_{2}$ will run on the remainder $v$ of the
input. The contribution of this branch of the computation to
$f_{\mathcal{G}}(w)$ is just the product of the value
$f_{\mathcal{G}_{2}}(v)$ and the coefficient
$f_{\mathcal{G}_{1}}(u)$, and will be positive if  both substrings are
accepted by the respective machines, and zero otherwise. Since
$\mathcal{G}_{2}$ starts running in this manner in each step, its part
of the overall state vector contains in general the sum of many $ 1
\times n_{2}  $ vectors, multiplied by their respective coefficients, at
any intermediate step of the computation. The cases where the empty
string appears in $L_{1} $ or $L_{2} $ are handled appropriately.

In other words,
     \[ f_{\mathcal{G}}(w)=\sum_{uv=w}
f_{\mathcal{G}_{1}}(u)f_{\mathcal{G}_{2}}(v), \]
and $ \mathcal{G} $ recognizes the concatenation of  $L_{1} $ with $L_{2} $.
\end{proof}
~~
\begin{theorem}
      \label{theorem:negative-one-sided-not-close-concatenation}
      S$ ^{=} $ is not closed under concatenation.
\end{theorem}
\begin{proof}
      $ L_{eq} $ and $ \{b\}^{+} $ are in S$ ^{=} $,
      but $ L_{eq \cdot b} = L_{eq}.\{b\}^{+} $ is not, due to Theorem \ref{theorem:L_twosided-3}.
\end{proof}
~~
\begin{theorem}
     \label{theorem-one-sided-closure-star}
     S$ ^{\neq} $ is closed under Kleene closure.
\end{theorem}
\begin{proof}
     If $ L \in $ S$ ^{\neq} $, then there exists a GPFA
     $ \mathcal{G}=(S,\Sigma,\{\mathsf{A}_{\sigma \in \Sigma}
\},\mathsf{v}_{0},\mathsf{f}) $
     such that $ L \cup \{\varepsilon\} $ is recognized by $
\mathcal{G} $ with one-sided
     cutpoint $ 0 $. Let $ n $ be the size of the state set $ S $.

     We construct a new GPFA $ \mathcal{G}^{'}=
     (S,\Sigma,\{\mathsf{A}_{\sigma \in \Sigma}^{'}
\},\mathsf{v}_{0},\mathsf{f}) $
     recognizing $ L^{*} $ (Kleene closure of $ L $) with one-sided
cutpoint $ 0 $.
     Each element of $ \{\mathsf{A}_{\sigma \in \Sigma}^{'} \} $ is defined as
\begin{equation}
      \label{equation:A_sigma2}
      \mathsf{A}_{\sigma}^{'}=\mathsf{A}_{\sigma}+\mathsf{X}_{\sigma},
\end{equation}
where $ \mathsf{X}_{\sigma} $ is an $ n \times n $ matrix, defined as
\begin{equation}
      \label{equation:X_sigma2}
      \mathsf{X}_{\sigma}=
                     \left(
                             \begin{array}{c|c|c|c}
                                      &  & ~~~~~~~~ &  \\
                                     \mathsf{v}_{0}(1)
\mathsf{A}_{\sigma}\mathsf{f} ~&~
                                     \mathsf{v}_{0}(2)
\mathsf{A}_{\sigma}\mathsf{f} ~&~
                                     \cdots &~
                                     \mathsf{v}_{0}(n)
\mathsf{A}_{\sigma}\mathsf{f} \\

\underbrace{~~~~~~~~~~~~}_{\mbox{column 1}} &
\underbrace{~~~~~~~~~~~~}_{\mbox{column 2}} &  &
                              \underbrace{~~~~~~~~~~~~}_{\mbox{column $ n $}}
                             \end{array}
                     \right).
\end{equation}
     For a given input string $ w \in \Sigma^{*} $, $ |w|=l $,
     \[ f_{\mathcal{G}^{'}}(w)= v_{0} [\mathsf{A}_{w_{1}} +
\mathsf{X}_{w_{1}}]
             [\mathsf{A}_{w_{2}} + \mathsf{X}_{w_{2}}] \cdots
[\mathsf{A}_{w_{l}}
+ \mathsf{X}_{w_{l}}] \mathsf{f}, \]
     and so
     \[ f_{\mathcal{G}^{'}}(w)=\sum_{u_{1}u_{2}\cdots u_{k}=w} \left(
\prod_{i=1}^{k} f_{\mathcal{G}}(u_{i}) \right), \]
     where $ 1 \le k \le l $ and each $ u_{i} \in \Sigma^{*} $. Therefore,
if $ w $ can be divided,
     i.e., $ w=u_{1} \cdots u_{k} $, such that each $ u_{i} \in L $
($ f_{\mathcal{G}}(u_{i})>0 $),
     then $ f_{\mathcal{G}^{'}}(w) >0 $. On the other
hand, if there is no such division,
     then $ f_{\mathcal{G}^{'}}(w)=0 $.
\end{proof}
~~
\begin{lemma}
      \label{lemma:Leq-prime-equal-L-lt}
      Let $ L_{eq^{'}} = \{ w \in \{a,b\}^{*} \mid |w|_{a}+1=|w|_{b}  \} $.
      Then, $ L^{+}_{eq^{'}} = L_{lt} $.
\end{lemma}
\begin{proof}
      It is obvious that if $ w \in L^{+}_{eq^{'}} $, then $ w \in L_{lt} $.
      If $ w \in L_{lt} $, then $ \exists k \in \mathbb{Z}^{+} $ such that
$ |w|_{b} = |w|_{a} + k $.
      Then, there must exist $ k+1 $ indices, $ i_{0} = 0 < i_{1} < i_{2} <
\cdots < i_{k} = |w| $, such that
      each prefix of $ w $ of length $ i_{j} $ has $ j $ more $ b $'s than
$ a $'s, where $ 1 \le j \le k $.
      In other words, $ w $ can be partitioned into $ k $ consecutive
substrings, $ u_{1}, u_{2}, \cdots, u_{k} $,
      satisfying
      \begin{enumerate}
              \item $ w = u_{1} u_{2} \cdots u_{k} $, and,
              \item $ |u_{j}| = i_{j} - i _{j-1} $, that is, $ u_{j} $
              begins with the $ (i_{j-1}+1) $th symbol of $ w $ and
ends with the
$ i_{j} $th symbol of $ w $,
              where $ 1 \le j \le k $.
      \end{enumerate}
      Since each $ u_{j} $ is a member of $ L_{eq^{'}} $, we can conclude
that $ w \in L^{+}_{eq^{'}} $.
\end{proof}
~~
\begin{theorem}
      \label{theorem:negative-one-sided-not-close-Kleene}
      S$ ^{=} $ is not closed under Kleene closure.
\end{theorem}
\begin{proof}
      It can be  shown easily that $ L_{eq^{'}} $ is in S$ ^{=} $.
      However, $ L_{eq^{'}}^{*} $,
      which is $ L_{lt} \cup \{ \varepsilon \} $ by Lemma
\ref{lemma:Leq-prime-equal-L-lt},
      is not in S$ ^{=} $ due to Corollary \ref{corollary:L_twosided-2}.
\end{proof}
~~
\begin{lemma}
     \label{lemma:homomorphism1}
     Let $ h: \Sigma \rightarrow \Sigma \setminus \{\kappa\} $ be a
homomorphism such that
     \[
             h(\sigma)=\left\lbrace
                     \begin{array}{ll}
                             \sigma & ,  \sigma \neq \kappa \\
                             \varepsilon &, \sigma=\kappa
                     \end{array}
             \right.,
     \]
     where $ \kappa $ is a specific symbol in $ \Sigma $.
     If $ L \subseteq \Sigma^{*} $ is in S$ ^{\neq} $, then so is $ h(L) $.
\end{lemma}
\begin{proof}
     Let $ \mathcal{G}=(S,\Sigma,\{\mathsf{A}_{\sigma \in \Sigma}
\},\mathsf{v}_{0},\mathsf{f}) $ be the GPFA
     recognizing $ L $ with one-sided cutpoint $ 0 $, and let $
\Sigma^{'}=\Sigma \setminus \{\kappa\} $.
     For any $ w \in h(L) $, there exists a $ u \in L $, such that $
h(u)=w $, i.e.,
     \[ u=\kappa^{c_{0}}w_{1}\kappa^{c_{1}}\cdots \kappa^{c_{|w|-1}}
w_{|w|} \kappa^{c_{|w|}} \]
     for some nonnegative integer $ c_{i} $'s, where $ 0 \le i \le |w| $.
     In fact, we can bound all $ c_{i}$'s by a natural number, say $ n_{L} $,
     due to Fact \ref{fact:dieu}: Suppose that none of the strings in
            \[ \{ \kappa^{c_{0}}w_{1}\kappa^{c_{1}}\cdots
\kappa^{c_{|w|-1}} w_{|w|} \kappa^{c_{|w|}}
      \mid 0 \le i \le |w|, 0 \le c_{i} \le n_{L} \}   \]
      are  members of $ L $
     (so they are all members of $ \overline{L} \in $ S$ ^{=} $).

     Then, by using Fact \ref{fact:dieu},
     \[  \kappa^{*}w_{1}\kappa^{c_{1}}\cdots \kappa^{c_{|w|-1}}
w_{|w|} \kappa^{c_{|w|}}
        \subseteq \overline{L}, (1 \le i \le |w|, 0 \le c_{i} \le n_{L}) \]
     \[  \kappa^{*}w_{1}\kappa^{*}\cdots \kappa^{c_{|w|-1}} w_{|w|}
\kappa^{c_{|w|}}
    \subseteq \overline{L}, (2 \le i \le |w|, 0 \le c_{i} \le n_{L})    \]
      \[ \vdots \]
      \[  \kappa^{*}w_{1}\kappa^{*}\cdots \kappa^{*} w_{|w|}
\kappa^{*}  \subseteq \overline{L}.  \]
     We conclude that $ w \notin h(L) $, which is a contradiction.

     Therefore, for any input string $ w \in (\Sigma^{'})^{*} $, we
can simulate the computation of $ \mathcal{G} $
     on some $ u $'s, where each $ c_{i} $ is guessed
nondeterministically from the set $ \{0,1,\cdots,n_{L}-1\} $.
     The following matrix can be defined to implement the
nondeterministic branching of the computation:
     \[ \mathsf{X}_{\kappa}=I+\sum_{j=1}^{n_{L}-1} \mathsf{A}_{\kappa}^{j}. \]
     By embedding $ \mathsf{X}_{\kappa} $ in a convenient way in the
definition of $ \mathcal{G} $,
     we can get the GPFA
     \[ \mathcal{G}^{'}=(S,\Sigma^{'},
     \{\mathsf{A}_{\sigma \in
\Sigma^{'}}^{'}=\mathsf{A}_{\sigma}\mathsf{X}_{\kappa}\},
     \mathsf{v}_{0}\mathsf{X}_{\kappa},\mathsf{f}), \]
     which recognizes $ h(L) $ with one-sided cutpoint $ 0 $. Hence, $
f_{\mathcal{G}^{'}}(w) $ can be calculated as
     \[ f_{\mathcal{G}^{'}}(w)=
     \sum_{u \in \{ \kappa^{c_{0}}w_{1}\kappa^{c_{1}}\cdots
\kappa^{c_{|w|-1}} w_{|w|} \kappa^{c_{|w|}} \} }
     f_{\mathcal{G}}(u) \]
     for the input string $ w \in (\Sigma^{'})^{*} $, where $ 0 \le c_{i}
\le n_{L}-1 $, and $ 0 \le i \le |w| $.
     Since the computation paths resulting in $ u \notin L $ produce $
f_{\mathcal{G}}(u)=0 $,
     $ f_{\mathcal{G}^{'}}(w)>0 $ is satisfied only when there is a
computation path resulting in $ u \in L $.
\end{proof}
~~
\begin{lemma}
     \label{lemma:homomorphism2}
     Let $ h : \Sigma \rightarrow \Upsilon^{*} $ be a homomorphism such
that $ |h(\sigma)|>0 $
     for all $ \sigma \in \Sigma $. If $ L \subseteq \Sigma^{*} $ is in S$
^{\neq} $, then so is $ h(L) $.
\end{lemma}
\begin{proof}
     Let $ \mathcal{G}=(S,\Sigma,\{\mathsf{A}_{\sigma \in \Sigma}
\},\mathsf{v}_{0},\mathsf{f}) $ be the GPFA
     recognizing $ L $ with one-sided cutpoint $ 0 $. We will show that
there exists a GPFA
     \[ \mathcal{G}^{'}=(S^{'},\Upsilon,\{\mathsf{A}_{\gamma \in
\Upsilon}^{'} \},\mathsf{v}_{0}^{'},\mathsf{f}^{'}) \]
     recognizing $ h(L) $ with one-sided cutpoint $ 0 $.

$ \mathcal{G}^{'} $ runs $ \mathcal{G} $ on a nondeterministically chosen input
$u = u_{1}u_{2}\cdots u_{|u|} \in \Sigma^{*} $, while checking whether its
own input string $ w $ matches $ h(u_{1})h(u_{2}) \cdots h(u_{|u|}) $
or not. For each such nondeterministic computation path, we
have the following cases:
     \begin{enumerate}
             \item at least one of the matches fail, $ h(u) \neq w $, then
             all entries representing the state vector of  $
\mathcal{G} $ in this branch
             are set to zero,
             \item all substitutions succeed, $ h(u)=h(u_{1})h(u_{2})\cdots
h(u_{|u|})=w $, then
             \begin{enumerate}
                     \item $ f_{\mathcal{G}}(u) = 0 $ for $ u \notin L $, and
                     \item $ f_{\mathcal{G}}(u) > 0 $ for $ u \in L $.
             \end{enumerate}
     \end{enumerate}
     $ f_{\mathcal{G}^{'}}(w) $ will again be defined as the summation
over all computation paths,
     i.e., $ f_{\mathcal{G}^{'}}(w)=\sum_{\{u \mid h(u)=w \} }
f_{\mathcal{G}}(u) $.
     Hence,  $ w \in h(L) $ only if there is at least one successful
substitution, $ f(u)=w $, and $ u \in L $.

     The technical details of $ \mathcal{G}^{'} $ are as follows:
     \begin{enumerate}
             \item For each $ \sigma \in \Sigma $, we will use a separate $ 1
\times |h(\sigma)||S| $ dimensional region in the
             state vector to trace the substitutions.
             \begin{equation}
                     \label{equation:longstatevector}
                     \left( \underbrace{(\cdots~~\cdots)}_{\sigma_{1}}
\underbrace{(\cdots~~\cdots)}_{\sigma_{2}} \cdots
                     \underbrace{(\cdots~~\cdots)}_{\sigma_{|\Sigma|}}\right)
             \end{equation}
             It is easily formulated that $ |S^{'}|=|S|\left( \sum_{\sigma \in
\Sigma} |h(\sigma)| \right) $.
             \item Each $ \mathsf{A}_{\gamma \in \Upsilon}^{'} $ is
defined with
respect to the separation above.
             \begin{equation}
                     \label{equation:Agamma}
                     \mathsf{A}_{\gamma}^{'}=\left(
                     \begin{array}{c|c|c|c}
                             \mathsf{A}_{\gamma,\sigma_{1}}^{'} & T
& \cdots & T \\ \hline
                             T & \mathsf{A}_{\gamma,\sigma_{2}}^{'}
& \cdots & T \\ \hline
                             \vdots & \vdots & \ddots & \vdots \\ \hline
                             T & T & \cdots &
\mathsf{A}_{\gamma,\sigma_{|\Sigma|}}^{'}
                     \end{array}
                     \right)
             \end{equation}
             Each $ T $ is almost a zero matrix, except a case which will be
described below.
             \item Let $ \sigma \in \Sigma $.
             Suppose that $ h(\sigma)=\gamma_{1}\gamma_{2} \cdots
\gamma_{l} \in
\Upsilon^{*} $ and $ l>0 $.
             Then, the
region corresponding to $ \sigma $ in $ \mathcal{G}^{'} $'s state vector
 can be
partitioned to $l$ $ 1 \times |S| $ blocks:
             \begin{equation}
                     \label{equation:statevector}
                     ( \underbrace{\cdots}_{1} \mid
\underbrace{\cdots}_{2} \mid \cdots \mid
                     \underbrace{\cdots}_{l} ).
             \end{equation}
             For $ \gamma \in \Upsilon $, $
\mathsf{A}_{\gamma,\sigma}^{'} $ can
be partitioned into blocks of dimension $ |S| \times |S| $:
             If $ l=1 $,
             \begin{equation}
                     \mathsf{A}^{'}_{\gamma,\sigma} = \left(T_{l}\right),
             \end{equation}
             and if $ l>1 $,
             \begin{equation}
                     \label{equation:Agammasigma}
                     \mathsf{A}^{'}_{\gamma,\sigma}=
                     \left(
                             \begin{array}{c|c|c|c|c}
                                     0 & T_{1} & 0 & \cdots & 0 \\ \hline
                                     0 & 0 & T_{2} & \cdots & 0 \\ \hline
                                     \vdots & \vdots & \vdots &
\ddots & \vdots \\ \hline
                                     0 & 0 & 0 & \cdots & T_{l-1} \\ \hline
                                     T_{l} & 0 & 0 & \cdots & 0
                             \end{array}
                     \right),
             \end{equation}
             where for $ \gamma_{i} = \gamma $, $ T_{i}=I $,  and  $ T_{i}=0 $
otherwise, $(0 < i < l)$;
             for $ \gamma_{l}=\gamma $, $ T_{l}=\mathsf{A}_{\sigma} $, and $
T_{l}=0 $ otherwise.
             Additionally, for $ \gamma_{l}=\gamma $, the
bottom-leftmost blocks
of all $ T $'s that are on the same row
             with $ \mathsf{A}^{'}_{\gamma,\sigma} $ in
(\ref{equation:Agamma})
are equal to $ \mathsf{A}_{\sigma} $;
             for $ \gamma_{l} \neq \gamma $,  all those blocks
contain all $ 0 $'s.
             \item
             \begin{equation}
                     \label{equation:v0longstatevector}
                     \mathsf{v}_{0}^{'}=
                     \left(
                     \underbrace{(\mathsf{v}_{0}\mid 0, \cdots,
0)}_{\sigma_{1}}
                     \underbrace{(\mathsf{v}_{0}\mid 0, \cdots,
0)}_{\sigma_{2}}  \cdots
                     \underbrace{(\mathsf{v}_{0}\mid 0, \cdots,
0)}_{\sigma_{|\Sigma|}} \right),
             \end{equation}
             \begin{equation}
                     \label{equation:flongstatevector}
                     \mathsf{f}^{'}=
                     \left(
                     \underbrace{(\mathsf{f}\mid 0, \cdots, 0)}_{\sigma_{1}}
                     \underbrace{(\mathsf{f}\mid 0, \cdots,
0)}_{\sigma_{2}}  \cdots
                     \underbrace{(\mathsf{f}\mid 0, \cdots,
0)}_{\sigma_{|\Sigma|}}
\right)^{\mathtt{T}}.
             \end{equation}
     \end{enumerate}
     When $ \sigma \in \Sigma $ is being guessed, where
     $ h(\sigma)=\gamma_{1}\gamma_{2} \cdots \gamma_{l} \in \Upsilon^{*} $
and $ l>0 $,
     the simulated state vector of $ \mathcal{G} $ is written in the
first slot of (\ref{equation:statevector}).
     Whenever the matches succeed for $ \gamma_{1}, \cdots, \gamma_{l-1} $,
     the simulated state vector of $ \mathcal{G} $ is transferred to
the next slot in
(\ref{equation:statevector}); in any other case, it is set to 0.
     When $ \gamma_{l} $ is successfully substituted, this simulated
state vector is
updated as if $ \mathcal{G} $ has read the symbol $ \sigma $,
     and the result is transferred anew in the first slots of all
the regions corresponding to symbols in $ \Sigma $;
     otherwise, it is set to 0.
\end{proof}
~~
\begin{theorem}
     S$ ^{\neq} $ is closed under homomorphism.
\end{theorem}
\begin{proof}
     Let $ h : \Sigma \rightarrow \Upsilon^{*} $ be a homomorphism,  $ L
\subseteq \Sigma^{*} $, and $L \in $ S$ ^{\neq} $.
     If $ h $ is a homomorphism of the form in Lemma
\ref{lemma:homomorphism2}, then the proof is complete.

     Otherwise, suppose that there are $ k \ge 1 $ symbols in $
\Sigma $, i.e., $
\sigma_{1}, \sigma_{2}, \cdots, \sigma_{k} $,
     such that $ h(\sigma_{i})=\varepsilon $, where $ 1 \le i \le k $. So,
we can define $ k $ homomorphisms of the form in
     Lemma \ref{lemma:homomorphism1}:
     \begin{eqnarray*}
             h_{1} & : & \Sigma \rightarrow \Sigma_{1},  \mbox{ where }
\Sigma_{1}=\Sigma \setminus \{\sigma_{1}\} \\
             h_{2}& : &  \Sigma_{1} \rightarrow \Sigma_{2},  \mbox{ where }
\Sigma_{2}=\Sigma_{1} \setminus \{\sigma_{2}\} \\
              & & \vdots \\
             h_{k} & : & \Sigma_{k-1} \rightarrow \Sigma_{k}, \mbox{ where }
\Sigma_{k}=\Sigma_{k-1} \setminus \{\sigma_{k}\}.
     \end{eqnarray*}
     Additionally, we define $ h_{k+1} : \Sigma_{k} \rightarrow
\Upsilon^{*} $, where $ h_{k+1}(\sigma)=h(\sigma) $
     for $ \sigma \in \Sigma_{k} $. Since $ h $ is the composition of the $
h_{i} $'s ($ 1 \le i \le k+1 $),
     $ h(L) $ is also in S$ ^{\neq} $ (Lemma \ref{lemma:homomorphism1} and
\ref{lemma:homomorphism2}).
\end{proof}
~~
\begin{theorem}
      \label{theorem:negative-one-sided-not-close-homomorphism}
      S$ ^{=} $ is not closed under homomorphism.
\end{theorem}
\begin{proof}
      Consider the languages $ L_{1}=\{w_{1}cw_{2} \mid w_{1} \in L_{eq}, w_{2} \in b^{+} \} $
      and $ L_{2}=\{w_{1}cw_{2} \mid w_{1} \in L_{eq}, w_{2} \in b^{*} \} $. 
      It is not hard to show that both languages are in S$ ^{=} $.
      Let $ h_{1} $ and $ h_{2} $ be two homomorphisms defined as
      \begin{itemize}
              \item $ h_{1}(a)=a $, $ h_{1}(b)=b $, $ h_{1}(c)=\varepsilon $, and
              \item $ h_{2}(a)=a $, $ h_{2}(b)=b $, $ h_{2}(c)=b $.
      \end{itemize}
      $ L_{eq \cdot b}=h_{1}(L_{1})=h_{2}(L_{2}) $, and so S$ ^{=} $ is not closed under
      ($ \varepsilon $-free) homomorphism due to Theorem
\ref{theorem:L_twosided-3}.
\end{proof}

\subsection{Common closure properties of S$ ^{\neq} $ and S$ ^{=} $}
~~

\begin{theorem}
      \label{theorem:close-inverse-homomorphism}
     S$ ^{\neq} $ and S$ ^{=} $ are closed under inverse homomorphism.
\end{theorem}
\begin{proof}
     Let $ h : \Sigma \rightarrow \Upsilon^{*} $ be a homomorphism, $ L
\subseteq \Upsilon^{*} $, and $L \in $ S$ ^{\neq} $,
     such that the GPFA $ \mathcal{G}=(S,\Upsilon,\{\mathsf{A}_{\gamma \in
     \Upsilon} \},\mathsf{v}_{0},\mathsf{f}) $
     recognizes $ L $ with one-sided cutpoint $ 0 $. It is easily
verified that
     \[ \mathcal{G}^{'}=(S,\Sigma,\{ \mathsf{A}_{\sigma \in
\Sigma}^{'} \},\mathsf{v}_{0},\mathsf{f}), \]
      where
      \[  \mathsf{A}_{\sigma \in \Sigma}^{'} = \left\lbrace
              \begin{array}{lll}

\mathsf{A}_{u_{1}}\cdots\mathsf{A}_{u_{|h(\sigma)|}} & , &
h(\sigma)=u_{1} \cdots u_{|h(\sigma)|}  \neq \varepsilon \\
                      I & , & h(\sigma)=\varepsilon
              \end{array}
      \right.,\]
     recognizes $ h^{-1}(L) $
     with one-sided cutpoint $ 0 $. The same setup can be extended
to any language in S$ ^{=} $.
\end{proof}
~~
\begin{theorem}
      \label{theorem:close-reversal}
    S$ ^{\neq} $ and S$ ^{=} $ are closed under reversal.
\end{theorem}
\begin{proof}
    We use the same idea as \cite{Tu69}.
    If $ L \in $  S$ ^{\neq} $, then there exists a GPFA
    $ \mathcal{G}=(S,\Sigma,\{\mathsf{A}_{\sigma \in \Sigma}
\},\mathsf{v}_{0},\mathsf{f}) $
    such that $ L $ is recognized by $ \mathcal{G} $ with one-sided
cutpoint 0. It is easily seen that
    $ \mathcal{G}^{'}=(S,\Sigma,\{\mathsf{A}_{\sigma \in
\Sigma}^{\mathtt{T}} \},
    \mathsf{f}^{\mathtt{T}},\mathsf{v}_{0}^{\mathtt{T}}) $
recognizes the reverse of $ L $
    with one-sided cutpoint $ 0 $. The same setup can be extended to
any language in S$ ^{=} $.
\end{proof}
~~
\begin{theorem}
      \label{theorem:close-word-quotient}
     S$ ^{\neq} $ and S$ ^{=} $  are closed under word quotient.
\end{theorem}
\begin{proof}
     If $ L \in $ S$ ^{\neq} $, then there exists a GPFA
     $ \mathcal{G}=(S,\Sigma,\{\mathsf{A}_{\sigma \in \Sigma}
\},\mathsf{v}_{0},\mathsf{f}) $
     such that $ L $ is recognized by $ \mathcal{G} $
with one-sided
     cutpoint $ 0 $. For any given $ w \in \Sigma^{*} $,
     \begin{enumerate}
             \item GPFA $
\mathcal{G}_{1}=(S,\Sigma,\{\mathsf{A}_{\sigma \in \Sigma} \},
             \mathsf{v}_{0}\mathsf{A}_{w_{1}}\cdots\mathsf{A}_{w_{|w|}},\mathsf{f})
$
                     recognizes the language $ \{y \mid wy \in L \} $
with one-sided cutpoint $ 0 $;
             \item GPFA $
\mathcal{G}_{2}=(S,\Sigma,\{\mathsf{A}_{\sigma \in \Sigma} \},
             \mathsf{v}_{0},\mathsf{A}_{w_{1}}\cdots\mathsf{A}_{w_{|w|}}\mathsf{f})
$
                     recognizes the language $ \{z \mid zw \in L \} $
with one-sided cutpoint $ 0 $.
     \end{enumerate}
     The same setup can be extended to any language in S$ ^{=} $.
\end{proof}
~~
\begin{theorem}
      \label{theorem:close-difference}
     S$ ^{\neq} $ and S$ ^{=} $ are not closed under difference.
\end{theorem}
\begin{proof}
    There exists an $ L \in $ S$ ^{\neq} $ such that $ \overline{L}
\notin $ S$ ^{\neq} $ .
    \begin{enumerate}
            \item $ \Sigma^{*} $ and $ L $ are in S$ ^{\neq} $,
             but $ \Sigma^{*} \setminus L = \overline{L} $ is not;
            \item $ \Sigma^{*} $ and $ \overline{L} $ are in S$ ^{=} $,
             but $ \Sigma^{*} \setminus \overline{L} = L $ is not.
    \end{enumerate}
\end{proof}
~~
\begin{theorem}
      \label{theorem:close-difference-with-regular-language}
    S$ ^{\neq} $ and S$ ^{=} $ are closed under difference with a
regular language.
\end{theorem}
\begin{proof}
    Regular languages are closed under complementation, and S$ ^{\neq} $
and S$ ^{=} $ are closed under intersection.
\end{proof}

For completeness, we list below the following easy facts about MCL and
NMCL:
\begin{enumerate}
      \item NMCL is closed under both union and intersection.
      \item Neither MCL nor NMCL is closed under
complementation \cite{BC01B}.
      \item Both MCL and NMCL are closed under inverse homomorphism
\cite{MC00}.
      \item Both MCL and NMCL are closed under word quotient
\cite{BP02}.
\end{enumerate}

\section{Concluding Remarks} \label{section:ConcludingRemarks}
In this paper, we gave a full characterization of the class of
languages recognized by all QFA models which are
at least as powerful as the Kondacs-Watrous QFA with cutpoint 0. This
is the only known case where the language recognition power of one-way
QFA's has been proven to be strictly
greater than that of their probabilistic counterparts\footnote{From a
``pedagogical" point of view, this seems to us to be
one of the simplest setups in which a quantum computational model can be
demonstrated to outperform the corresponding probabilistic
model.}. The superiority of QFA's over PFA's with regard
to space efficiency in this setting was demonstrated.
We also examined the limitations of recognition with one-sided error
for these models.
Several new closure properties of the related classes $ S^{\neq} $ and
$ S^{=} $ were proven.

The relationship between nondeterministic quantum complexity classes
and counting classes has been studied in detail.
It is known that NQP = co-C$ _{=} $P \cite{YY99}.
More relevantly for our work, Watrous \cite{Wa99} has shown that
NQSPACE($ s $) = co-C$ _{=} $SPACE($ s $) for
$ s = \Omega(\log(n)) $.
Note that the
subset of S$ ^{\neq} $ defined using PFA's that only contain efficiently
computable transition probabilities and cutpoint $ \frac{1}{2} $
equals co-C$ _{\mbox{=}} $SPACE(1),
so we have proven\footnote{All our proofs stand 
when the transition probabilities and amplitudes are restricted to be
efficiently computable numbers, as mentioned in \cite{BV97}.
We can in fact prove that the collection of languages recognized by
the most general model of NQFA's \cite{Pa00} is precisely the class
S$ ^{\neq}_{rat} $ (see Definition \ref{definition:SLeq-rat}) when all
the amplitudes of the NQFA are
restricted to be rational numbers.} that $ \mbox{co-C} _{=} $SPACE(1) $
\subseteq $ NQSPACE(1),
and whether the inclusion is strict or not depends on
whether a two-way head would increase the computational power of a
NQFA\footnote{For any two-way PFA $ \mathcal{M} $ and cutpoint $
\lambda_{1} \in [0,1) $,
there exist a one-way PFA $ \mathcal{P} $ and a cutpoint $ \lambda_{2}
\in [0,1) $ such that
$ (\mathcal{M},\lambda_{1}) \equiv (\mathcal{P},\lambda_{2}) $
\cite{Ka89}, whereas two-way QFA's are more powerful than one-way
QFA's in the general unbounded error setting \cite{YS09D}.}.

The superiority of NQFA's over classical NFA's has ramifications about
relationships among classical
and quantum nondeterministic space complexity classes for all
sublogarithmic bounds.
Although the following result follows from a combination of previously
known facts,
we have not seen it stated anywhere:
\\
\begin{theorem}
       NSPACE($ s $) $ \subsetneq $ NQSPACE($ s $) for $ s=o(\log(n)) $.
\end{theorem}
\begin{proof} (Sketch.)
       Quantum Turing machines can simulate probabilistic Turing machines
easily for any common space bound \cite{Wa09}.
       There exists a NQFA with efficiently computable amplitudes
       (i.e. a constant-space nondeterministic quantum Turing machine) which
recognizes the language
       $ L_{neq}=\{w\in\{a,b\}^{*} \mid |w|_{a} \neq |w|_{b}\} $ \cite{BC01B}.
       It is easily seen that $ L_{neq} $ is a nonregular deterministic
context-free language (DCFL).
       It is known that no nonregular DCFL is in NSPACE($ s $) for $
s=o(\log(n)) $ \cite{AGM92}.
\end{proof}
For space bounds $ s \in \Omega(\log(n)) $,
       all we know in this regard is the trivial fact that NSPACE($ s $) $
\subseteq $ NQSPACE($ s $) \cite{Wa99}.

The succinct QFA models alluded to in Section 3.3 form the basis of a
demonstration \cite{YS09F} of the fact that two-way QFA's can have a
similar state complexity advantage over both their one-way versions,
and two-way classical nondeterministic automata.

One important QFA variant that was not considered in this paper is the
Latvian QFA \cite{ABGKMT06}, which is a generalization of the MCQFA not thought
to be as powerful as the KWQFA. An examination of the corresponding
classes for this model would be interesting.

Some other open questions related to this work are listed below.
\\
\begin{openproblem}
      Is MCL closed under union? Intersection?
\end{openproblem}
~~
\begin{openproblem}
      Do NMCL and MCL coincide?
\end{openproblem}
~~
\begin{openproblem}
      Does S$ ^{\neq} \cap $ S$ ^{=} $ contain a nonregular language?
\end{openproblem}
~~
\begin{openproblem}
      Is S$ ^{\neq} $ countable or uncountable?
\end{openproblem}
~~
\begin{openproblem}
      Is S$ ^{>} $ closed under complementation? (page 158 of \cite{Pa71})
\end{openproblem}
~~
\begin{openproblem}
      Is S$ ^{=} $ a subset of S$ ^{>} $? (page 173 of \cite{Pa71})
\end{openproblem}
~~
\begin{openproblem}
Can NQFA's with a two-way tape head recognize more
languages than the one-way model discussed here?
\end{openproblem}
\section*{Acknowledgments} \label{section:Acknowledgment}
We are grateful to Andris Ambainis, John Watrous, and Flavio
D'Alessandro for their helpful
comments on the subject matter of this paper. We also thank
R\={u}si\c{n}\v{s} Freivalds for kindly providing us a copy of
reference \cite{LSF09}.

\bibliographystyle{plain}
\bibliography{YakaryilmazSay}

\begin{thebibliography}{10}

\bibitem{AGM92}
Helmut Alt, Viliam Geffert, and Kurt Mehlhorn.
\newblock A lower bound for the nondeterministic space complexity of
  context-free recognition.
\newblock {\em Information Processing Letters}, 42(1):25--27, 1992.

\bibitem{ABGKMT06}
Andris Ambainis, Martin Beaudry, Marats Golovkins, Arnolds \c{K}ikusts, Mark
  Mercer, and Denis Th{\'e}rien.
\newblock Algebraic results on quantum automata.
\newblock {\em Theory of Computing Systems}, 39(1):165--188, 2006.

\bibitem{AF98}
Andris Ambainis and R\={u}si\c{n}\v{s} Freivalds.
\newblock 1-way quantum finite automata: strengths, weaknesses and
  generalizations.
\newblock In {\em FOCS'98: Proceedings of the 39th Annual Symposium on
  Foundations of Computer Science}, pages 332--341, Palo Alto, California,
  1998.

\bibitem{AW02}
Andris Ambainis and John Watrous.
\newblock Two--way finite automata with quantum and classical states.
\newblock {\em Theoretical Computer Science}, 287(1):299--311, 2002.

\bibitem{BV97}
Ethan Bernstein and Umesh Vazirani.
\newblock Quantum complexity theory.
\newblock {\em SIAM Journal on Computing}, 26(5), 1997.

\bibitem{BC01B}
Alberto Bertoni and Marco Carpentieri.
\newblock Analogies and differences between quantum and stochastic automata.
\newblock {\em Theoretical Computer Science}, 262(1-2):69--81, 2001.

\bibitem{BMP03}
Alberto Bertoni, Carlo Mereghetti, and Beatrice Palano.
\newblock Quantum computing: 1-way quantum automata.
\newblock In Zolt\'{a}n \'{E}sik and Zolt\'{a}n F\"{u}l\"{o}p, editors, {\em
  Developments in Language Theory}, volume 2710 of {\em LNCS}, pages 1--20.
  Springer, 2003.

\bibitem{Br98}
Alex Brodsky.
\newblock Models and characterizations of 1-way quantum finite automata.
\newblock Master's thesis, The University of British Columbia, 1998.

\bibitem{BP02}
Alex Brodsky and Nicholas Pippenger.
\newblock Characterizations of 1--way quantum finite automata.
\newblock {\em SIAM Journal on Computing}, 31(5):1456--1478, 2002.

\bibitem{Bu68}
R.~G. Bukharaev.
\newblock Theory of probabilistic automata.
\newblock {\em Kibernetika}, (2):6--23, 1968.

\bibitem{Ci01}
Massimo~Pica Ciamarra.
\newblock Quantum reversibility and a new model of quantum automaton.
\newblock In {\em FCT '01: Proceedings of the 13th International Symposium on
  Fundamentals of Computation Theory}, pages 376--379, London, UK, 2001.
  Springer-Verlag.

\bibitem{Co69}
Stephen~N. Cole.
\newblock Real-time computation by n-dimensional iterative arrays of
  finite-state machines.
\newblock {\em IEEE Transactions on Computers}, 18(4):349--365, 1969.

\bibitem{DD06}
Flavio D'Alessandro and Alessandro D'Andrea.
\newblock A non-commutativity statement for algebraic quaternions.
\newblock {\em International Journal of Algebra and Computation},
  16(3):583--602, 2006.

\bibitem{Di71}
Phan~Dinh Di\^eu.
\newblock On a class of stochasic languages.
\newblock {\em Mathematical Logic Quarterly}, 17(1):421--425, 1971.

\bibitem{Fl73B}
Michel Fliess.
\newblock Automates stochastiques et s{\'e}ries rationnelles non commutatives.
\newblock In {\em Automata, Languages, and Programming}, pages 397--411, 1973.

\bibitem{Fl73A}
Michel Fliess.
\newblock Propri\'{e}t\'{e}s bool\'{e}ennes des langages stochastiques.
\newblock {\em Mathematical Systems Theory}, 7(4):353--359, 1973.

\bibitem{Hi07}
Mika Hirvensalo.
\newblock Improved undecidability results on the emptiness problem of
  probabilistic and quantum cut-point languages.
\newblock In {\em SOFSEM 2007: Theory and Practice of Computer Science}, volume
  4362 of {\em Lecture Notes in Computer Science}, pages 309--319. Springer
  Berlin / Heidelberg, 2007.

\bibitem{Hi08}
Mika Hirvensalo.
\newblock Various aspects of finite quantum automata.
\newblock In {\em DLT '08: Proceedings of the 12th international conference on
  Developments in Language Theory}, pages 21--33, Berlin, Heidelberg, 2008.
  Springer-Verlag.

\bibitem{Ka89}
J\={a}nis Ka{\c{n}}eps.
\newblock Stochasticity of the languages acceptable by two-way finite
  probabilistic automata.
\newblock {\em Diskretnaya Matematika}, 1:63--67, 1989.
\newblock (Russian).

\bibitem{KW97}
Attila Kondacs and John Watrous.
\newblock On the power of quantum finite state automata.
\newblock In {\em FOCS'97: Proceedings of the 38th Annual Symposium on
  Foundations of Computer Science}, pages 66--75, Miami, Florida, 1997.

\bibitem{LSF09}
Lelde L\={a}ce, Oksana Scegulnaja-Dubrovska, and R\={u}si\c{n}\v{s} Freivalds.
\newblock Languages recognizable by quantum finite automata with cut-point 0.
\newblock presented at the 35th International Conference on Current Trends in
  Theory and Practice of Computer Science, SOFSEM, 2009.

\bibitem{La74}
J\={a}nis Lapi\c{n}\v{s}.
\newblock On nonstochastic languages obtained as the union and intersection of
  stochastic languages.
\newblock {\em Avtom. Vychisl. Tekh.}, (4):6--13, 1974.
\newblock (Russian).

\bibitem{LQ08}
Lvzhou Li and Daowen Qiu.
\newblock Determining the equivalence for one-way quantum finite automata.
\newblock {\em Theoretical Computer Science}, 403(1):42--51, 2008.

\bibitem{LZ77}
Richard~J. Lipton and Yechezkel Zalcstein.
\newblock Word problems solvable in logspace.
\newblock {\em Journal of the ACM}, 24(3):522--526, 1977.

\bibitem{LS77}
Roger~C. Lyndon and Paul~E. Schupp.
\newblock {\em Combinatorial Group Theory}.
\newblock Springer-Verlag, 1977.

\bibitem{Ma93}
Ioan Macarie.
\newblock Closure properties of stochastic languages.
\newblock Technical report, University of Rochester, Rochester, NY, USA, 1993.

\bibitem{MPP01}
Carlo Mereghetti, Beatrice Palano, and Giovanni Pighizzini.
\newblock Note on the succinctness of deterministic, nondeterministic,
  probabilistic and quantum finite automata.
\newblock {\em Theoretical Informatics and Applications}, 35(5):477--490, 2001.

\bibitem{MC00}
Cristopher Moore and James~P. Crutchfield.
\newblock Quantum automata and quantum grammars.
\newblock {\em Theoretical Computer Science}, 237(1-2):275--306, 2000.

\bibitem{NIHK02}
Masaki Nakanishi, Takao Indoh, Kiyoharu Hamaguchi, and Toshinobu Kashiwabara.
\newblock On the power of non-deterministic quantum finite automata.
\newblock {\em IEICE Transactions on Information and Systems},
  E85-D(2):327--332, 2002.

\bibitem{Na99}
Ashwin Nayak.
\newblock Optimal lower bounds for quantum automata and random access codes.
\newblock In {\em FOCS'99: Proceedings of the 40th Annual Symposium on
  Foundations of Computer Science}, pages 369--376, Washington, DC, USA, 1999.
  IEEE Computer Society.

\bibitem{Pa00}
Kathrin Paschen.
\newblock Quantum finite automata using ancilla qubits.
\newblock Technical report, University of Karlsruhe, 2000.

\bibitem{Pa71}
Azaria Paz.
\newblock {\em Introduction to Probabilistic Automata}.
\newblock Academic Press, New York, 1971.

\bibitem{Ra92}
Bala Ravikumar.
\newblock Some observations on 2-way probabilistic finite automata.
\newblock In {\em Proceedings of the 12th Conference on Foundations of Software
  Technology and Theoretical Computer Science}, pages 392--403, London, UK,
  1992. Springer-Verlag.

\bibitem{Sw94}
S.~\'{S}wierczkowski.
\newblock A class of free rotation groups.
\newblock {\em Indagationes Mathematicae}, 5(2):221--226, 1994.

\bibitem{Tu68}
Paavo Turakainen.
\newblock On stochastic languages.
\newblock {\em Information and Control}, 12(4):304--313, 1968.

\bibitem{Tu69}
Paavo Turakainen.
\newblock Generalized automata and stochastic languages.
\newblock {\em Proceedings of the American Mathematical Society}, 21:303--309,
  1969.

\bibitem{Tu69B}
Paavo Turakainen.
\newblock On languages representable in rational probabilistic automata.
\newblock {\em Annales Academiae Scientiarum Fennicae, Ser.A}, (439):4--10,
  1969.

\bibitem{Tu71}
Paavo Turakainen.
\newblock Some closure properties of the family of stochastic languages.
\newblock {\em Information and Control}, 18(3):253--256, 1971.

\bibitem{Tu81}
Paavo Turakainen.
\newblock On nonstochastic languages and homomorphic images of stochastic
  languages.
\newblock {\em Information Sciences}, 24(3):229--253, 1981.

\bibitem{Wa99}
John Watrous.
\newblock Space-bounded quantum complexity.
\newblock {\em Journal of Computer and System Sciences}, 59(2):281--326, 1999.

\bibitem{Wa09}
John Watrous.
\newblock Quantum computational complexity.
\newblock In Robert~A. Meyers, editor, {\em Encyclopedia of Complexity and
  Systems Science}, pages 7174--7201. Springer, 2009.

\bibitem{YS09A}
Abuzer Yakary{\i}lmaz and A.~C.~Cem Say.
\newblock Language recognition by generalized quantum finite automata with
  unbounded error.
\newblock In {\em 4th Workshop on Theory of Quantum Computation, Communication,
  and Cryptography, TQC2009}, Waterloo, Ontario, Canada, 2009.

\bibitem{YS09D}
Abuzer Yakary{\i}lmaz and A.~C.~Cem Say.
\newblock Languages recognized with unbounded error by quantum finite automata.
\newblock In {\em CSR'09: Proceedings of the Fourth International Computer
  Science Symposium in Russia}, volume 5675 of {\em Lecture Notes in Computer
  Science}, pages 356--367, 2009.

\bibitem{YS09F}
Abuzer Yakary{\i}lmaz and A.~C.~Cem Say.
\newblock Succinctness of two-way probabilistic and quantum finite automata.
\newblock Technical Report arXiv:0903.0050v2, 2009.
\newblock A preliminary version of this paper was presented at the AutoMathA
  Plenary Conference 2009, in Li\`{e}ge, Belgium.

\bibitem{YS09G}
Abuzer Yakary{\i}lmaz and A.~C.~Cem Say.
\newblock Unbounded-error quantum computation with small space bounds.
\newblock in preparation, 2010.

\bibitem{YY99}
Tomoyuki Yamakami and Andrew Chi-Chih Yao.
\newblock $ \mbox{NQP}_{ \mathbb{C} } $ = co-$ \mbox{C}_{=}\mbox{P}$.
\newblock {\em Information Processing Letters}, 71(2):63--69, 1999.

\end{thebibliography}

\end{document}